\documentclass[12pt,english,leqno]{article}
\usepackage{mathptmx}
\usepackage[T1]{fontenc}
\usepackage[unicode=true,pdfusetitle,
 bookmarks=true,bookmarksnumbered=false,bookmarksopen=false,
 breaklinks=false,pdfborder={0 0 0},pdfborderstyle={},backref=false,colorlinks=true] {hyperref}
\hypersetup{
 citecolor=blue, linkcolor=blue, urlcolor=blue}
\usepackage{geometry}
\usepackage[utf8]{inputenc}
\usepackage{amsmath,mathtools,commath}
\usepackage{amsthm}
\usepackage{amssymb}
\usepackage[authoryear,round]{natbib}
\usepackage{hyperref}
\usepackage{amsfonts}
\usepackage{palatino}
\usepackage{color}
\usepackage{array}
\usepackage{verbatim}
\usepackage{multirow}
\usepackage{amsmath}
\usepackage{amsthm}
\usepackage{graphicx}
\usepackage{setspace}
\usepackage{hyperref}
\usepackage{epigraph}
\usepackage{subfigure}
\usepackage{accents}
\usepackage{changepage}

\newtheorem{theorem}{Theorem}
\newtheorem{lemma}{Lemma}
\newtheorem{proposition}{Proposition}

\theoremstyle{definition}
\newtheorem{definition}{Definition}
\theoremstyle{remark}
\newtheorem{remark}{Remark}
\newtheorem{example}{Example}
\newtheorem{corollary}{Corollary}
\theoremstyle{claim}

\theoremstyle{fact}
\newtheorem{fact}{Fact}

\newcommand{\myenddefinition}{$\blacklozenge$}
\newcommand* \supp {\mathop{\rm supp}\nolimits }
\def\ss{\mathfrak{s}}
\def\aa{\mathfrak{a}}
\def\dd{\mathfrak{d}}
\def\cc{\mathfrak{c}}
\def\sc{\mathfrak{sc}}
\def\sd{\mathfrak{sd}}
\def\ac{\mathfrak{ac}}
\def\ad{\mathfrak{ad}}
\def \hh{\mathfrak{h}}
\def \ee{\mathfrak{e}}
\def \rr{\mathfrak{r}}
\def \vv{\mathfrak{v}}
\def \he{\mathfrak{he}}
\def \hr{\mathfrak{hr}}
\def \hv{\mathfrak{hv}}
\def \de{\mathfrak{he}}
\def \dr{\mathfrak{hr}}
\def \dv{\mathfrak{dv}}
\def\R{\mathbb{R}}

\def\a{x}

\def\D{\Delta}
\def\e{\varepsilon}

\def \y{\mathfrak{y}}
\def\:={\coloneqq}
\def\INT{\textrm{Int}}
\def\Rec{\textrm{Rec}}

\begin{document}
\title{{Mentors and Recombinators:\\ Multi-Dimensional Social Learning}\thanks{We thank seminar and conference audiences at Learning, Evolution and Games conference (Lucca, 2022), Bar-Ilan University, the Hebrew University of Jerusalem, Technion (Israel Institute of Technology), Stony Brook University, UC Santa Barbara, UC San Diego, and the University of Manchester, for various helpful comments. 
SA gratefully acknowledges financial support from a Fine Fellowship provided by the Technion. 
All authors gratefully acknowledge research support from the Israel Science Foundation (\#448/22). 
}}

\author{Srinivas Arigapudi\thanks{Department of Economic Sciences, IIT Kanpur, India, 
email:\href{mailto:arigapudi@iitk.ac.in}{arigapudi@iitk.ac.in.}}
\and Omer Edhan\thanks{Department of Economics, University of Manchester, UK, email:\href{mailto:omeredhan.idan@manchester.ac.uk} {omeredhan.idan@manchester.ac.uk}}
\and Yuval Heller\thanks{Department of Economics, Bar-Ilan University, Israel,  
 email:\href{mailto:yuval.heller@biu.ac.il}{yuval.heller@biu.ac.il.}} \and Ziv Hellman\thanks{Department of Economics, Bar-Ilan University, Israel, email:\href{mailto:ziv.hellman@biu.ac.il}{ziv.hellman@biu.ac.il} }} 

\date {\today}
\maketitle

\begin{abstract}
    {
    We study games in which the set of strategies is multi-dimensional, and new agents might learn various strategic dimensions from  different mentors. 
    We introduce a new family of dynamics, the recombinator dynamics, which is characterised by a single parameter, the recombination rate $r\in[0,1]$. 
    The case of $r=0$ coincides with the standard replicator dynamics. The opposite case of $r=1$ corresponds to a setup in which each new agent learns each new strategic dimension from a different mentor.
    We fully characterise stationary states and stable states under these dynamics, and we show that they predict novel 
 and empirically-relevant behaviour in various applications.}
\end{abstract}

\section{Introduction}
\epigraph{\hspace{-7cm}Ben Zoma said: Who is wise? One who learns from every one,
as it is written\\
\hspace{-7cm} (Psalms 119:99): `From all of my teachers I have gained understanding'.
}{\textit{Ethics of the Fathers,\,\,\,  Chapter 4}}
The theory of the dynamics of behavioural traits in populations, developed in the context of imitative games, has been one of the most successful sub-fields of evolutionary game theory.
A rich set of results and insights regarding the convergence and the stability of equilibria of population-level traits forms the core of the subject. 
They have been applied to models of social interactive situations, economic models, and biological models of evolution and ecology.  

Our contribution here begins with the observation that in many of the replicator population game models in the literature to date the traits that are the main objects of study are dimensionless,
and player types are usually entirely identified with single traits.
This is far from realistic; each of us is composed of an ensemble of traits and behaviours, from the way we speak and dress to the subjects that we study and the professions that we choose.
In a broad sense, each of us may be identified with those trait ensembles, and social and professional success often depends on the entire suite of those traits.

The standard replicator paradigm of evolutionary game theory typically posits a population with a set of traits, while each agent in the population (which may be an individual, a firm, or a social or economic unit) bears only a single trait from that set. A two-player symmetric normal-form game exists, and a trait is associated with a population-dependent fitness value, representing the payoff a player with that trait would receive on average playing against the current distribution of traits in the population.

Here we extend this model to what we term the recombinator model by positing that there is a set of dimensions of cardinality $|D|$, where each dimension is itself a set of traits.
An agent is now associated with a $|D|$-tuple of traits, one from each dimension.
Such a tuple of traits is a type, and each type is assigned a fitness payoff.
 
A newly born agent can either sample a single mentor, as in the standard replicator model, and then imitate the entirety of that mentor's type for his own type, or alternatively be what we term a combinator, who independently samples $|D|$ incumbent agents, each of whom is a mentor.
For each $1 \le d \le |D|$, the combinator agent imitates the $d$-th trait of the $d$-th mentor.

As an example, consider a newly formed commercial enterprise. 
It may entirely emulate another successful company, copying every aspect of the mentor company's corporate structure and strategic practices; this would be replicator imitation.
Alternatively, in combinator imitation, the newly formed company would copy the corporate governance of one mentor company, the brand management of a second company, and the marketing and customer services approach of a third company.

For each $r \in [0,1]$, which we term the recombination rate, we obtain a recombinator dynamic, in which a newly born agent is a combinator with probability $r$ and a replicator with probability $1-r$.
In this way one obtains a recombinator law of motion equation; when $r=0$ the recombinator equation reduces to the standard replicator equation, and when $r=1$ a pure combinator model is attained. 
From here, we may inquire about trajectories, convergence, stability, and similar properties.

The recombinator model, with its tuple of traits comprising the type of an agent, was inspired by considering models of DNA and genetic recombination in biological evolution.
One of the central insights of the past century in biological evolution, the gene-centred view in which the interactions of genes are paramount over those of individuals for tracing long-term trajectories, is paralleled in the social learning setting of this paper, where the analogy is a traits-centred view.

We can identify two games occurring in parallel in our model: at a visible level, agents interact in the two-player normal form game $G_P$, which determines their population-dependent payoffs.
More subtly, there is also a game being played between the traits:
one can interpret each trait as competing against the other traits within its dimension, analogous to the competition between genetic alleles.

This view of parallel types and traits games is not simply a modelling convenience; it is helpful for fully analysing trajectories and asymptotic convergence results in recombinator models. 
For one thing, when $r>0$ the recombinator dynamics may violate payoff monotonicity, a property that is necessary for many of the standard results in imitative games. This can lead to stable states in which all agents have strictly dominated types. 

To remedy this we introduce what we term the $r$-payoff function, which combines the effects of the combinator and replicator components of the dynamics into a single vector field along which trajectories in the recombinator model flow.
Our first result (Proposition \ref{pro:stationary-states-recombinator}) shows that the dynamics are monotone with respect to these $r$-payoffs, and uses this to present our first ``if and only if'' characterisation of the set of stationary states.

At the stationary state of a convergent trajectory of the recombinator dynamic the surviving types may exhibit different payoffs.
In contrast, at the traits level, the dynamics reliably select for traits with higher payoffs (in ensemble with other traits), with lower performing traits eventually becoming extinct. At stationary state convergence, all surviving traits have exactly the same payoff. This is formalised in our second result (Proposition \ref{prop-trait-stationary}), which presents a second ``if and only if'' characterisation of stationary states in terms of the traits payoffs and the correlation in the population between different traits.

 It is well-known (see Facts \ref{claim-necc-asymp-replic}--\ref{claim-suffic-asymp-replic}) that a stationary state is asymptotically stable under the standard replicator dynamics essentially if and only if it satisfies: 
 \begin{enumerate}
     \item \emph{internal stability}: the payoff matrix restricted to the incumbent types is negative-definite, and
     \item \emph{external stability} the payoffs of external types  is lower than the incumbents. 
 \end{enumerate} 
 
 Our main results (Theorems \ref{thm:asymptotic-neccesary}--\ref{thm:asymptotic-sufficient}) extend this characterisation to the recombinator dynamics. 
 Specifically, we show that a stationary state is asymptotically stable under the combinator dynamics essentially\footnote{\ Asymptotic stability implies the weak variants of these 3 conditions, and it is implied by their strict counterparts.} if and only if
  \begin{enumerate}
      \item \emph{internal stability}: the $r$-Jacobian matrix (which replaces the standard payoffs with the $r$-payoffs described above) is negative-definite, 
      \item \emph{external stability against traits}: the payoff of the incumbents is higher than the payoff of external traits; we show that each small invasion of an external trait induces a unique distribution of types who carry this trait (partners), and this distribution is used to calculate the trait's payoff.
     \item  \emph{external stability against types}:
     the payoff of incumbents is  higher than than the payoff of external types  multiplied by $1-r<1$, where this factor $1-r$ reflects the fact that stability against external types become easier the higher the recombination rate $r$ (because traits of a successful external type recombine with the incumbent traits into hybrid types, which might have lower payoffs).
  \end{enumerate} 
  
We apply our dynamics in two simple applications. The first application is a prisoner's dilemma with partially enforceable contracts in which each player chooses (1) making either a costless enforceable promise or a costly non-enforceable promise to cooperate, and (2) either cooperating or defecting. All recombination rates admit an inefficient stable state in which all players make costly non-enforceable promises and defect. By contrast, if the recombination rate is not too low, the game also admits a Pareto-optimal stable state in which all agents make enforceable promises and cooperate.
Our second application (Example \ref{exam-emotional-HD}) is a hawk-dove game which is enriched by adding 
a second dimension with three traits: (1-2) two specialised traits: being emotional (rational), which is somewhat helpful for a hawkish (dovish) player, while  being harmful for a dovish (hawkish) player, and (3) being versatile (no effect on the game payoffs). We show that there is a recombination threshold, such that for any  recombination rate below (above) this threshold, the unique stable state involves specialised (versatile) traits.

Most of our paper focuses on a particular equation of motion expressed by the recombinator dynamic. In Section \ref{sec-extension} we extend our characterisation results for stationary and stable states to a general setup that captures a broader class of dynamics. 

Real-life social interactions, whether between individuals or economic units, are typically multi-dimensional, as opposed to the single-dimensional quality of many of the existing replicator dynamics in the literature.
As argued by \cite{arad2018multi} in games with a multi-dimensional strategy space, players tend to think in terms of strategy traits rather than the strategies themselves. 
We therefore expect that the recombinator dynamics
can serve in future research efforts as a tool for understanding observed behaviour in a large variety of applications. 

The paper is structured as follows. In Section \ref{sec-related}, we discuss the related literature. Section\ref{sec-model} presents our model. We characterise stationary states in Section \ref{sec-staionary} and asymptotically stable states in Section \ref{sec-stable}. Section \ref{sec-extension} extends our analysis to a broader set of dynamics. We conclude in Section \ref{sec-conclusion}. The technical proofs are presented in the Online Appendix.

\section{Related Literature}\label{sec-related}

Our research belongs to the evolutionary game theory literature (pioneered in the seminal paper of \citealp{maynard-smith1973logic}, with an earlier brief discussion in John Nash's unpublished dissertation, see \citealp{weibull1994mass}). 
This literature considers a game that is played over and over again by biologically or socially conditioned players who are randomly drawn from
large populations. Occasionally, new agents join the population (or incumbents revise their behaviour), and they learn how to play based on observing the (possibly noisy) behaviour and payoffs of some of the incumbent agents (see, \citealp{weibull1997evolutionary} and \citealp{sandholm2010population} for a textbook introduction, and \citealp{newton2018evolutionary} for a comprehensive recent survey of the literature). 

A commonly applied dynamic to capture how the aggregate behaviour gradually changes in such a learning process is the \emph{replicator dynamic} (\citealp{taylor1978evolutionary}) in which the relative (per capita) change in the proportion of agents playing each action $a$ (henceforth, $a$-agents) is proportional to the average payoff of the $a$-agents. Although the replicator dynamic was originally developed to describe natural selection in a genetic evolution, it has been successfully applied to many situations of social learning (examples for various applications of the replicator dynamic and its extensions include  \citealp{borgers1997learning,hopkins2002two, skyrms2004stag, cressman2014replicator, sawa2014evolutionary, mertikopoulos2018riemannian}). 
In particular, the replicator dynamic models imitative processes in which new agents imitate the behaviour of successful incumbents (mentors), with the probability that a specific mentor is chosen to be imitated proportional to that mentor's payoff (\citealp{bjornerstedt1994nash}).%
\footnote{\ Experimental evidence for the predictions induced by imitative dynamics (similar to the replicator dynamic) is presented in  \cite{oprea2011separating,cason2014cycles,hoffman2015experimental, benndorf2016equilibrium,benndorf2021games}.}

 The replicator dynamic is part of a broad family of dynamics that satisfy \emph{payoff monotonicity}: the relative growth of action $a$ is larger than the relative growth of action $a'$ if and only if $a$ receives higher payoff than $a'$.
  The so-called Folk Theorem of evolutionary game theory (see, e.g., \citealp{nachbar1990evolutionary,hofbauer2003evolutionary}) states that there are close relations between stable states in the family of payoff monotone dynamics and Nash equilibria, namely, that stable stationary points are Nash equilibria of the game, interior trajectories converge to Nash equilibria, and strict Nash equilibria are asymptotically stable. 
  \footnote{\ The folk theorem holds under a weaker property than monotonicity, namely, weak payoff positivity (\cite[Proposition 4.11]{weibull1994mass}). 
 Weak payoff positivity is the condition that if there exist some actions that yield strictly higher payoffs than the average payoff in the population, then at least one of these actions has a positive growth rate. 
 This weak assumption holds in many dynamics, such as better reply dynamics (\citealp{hart2002evolutionary,arieli2016stochastic}), and  best reply dynamics (\citealp{hwang2017payoff,babichenko2018fast,sawa2019evolutionary}). 
 The few models that violate Weak payoff positivity include action-sampling dynamics (\citealp{sandholm2001almost,oyama2015sampling,arigapudi2023heterogeneous}) and payoff-sampling dynamics (\citealp{sethi2000stability,sandholm2020stability,arigapudi2020instability}).} 
 
The vast majority of the evolutionary game theory literature assumes that the learning process is one-dimensional. 
In particular, most of the literature on imitative processes assumes that a new agent mimics the behaviour of only a single mentor.  In what follows we describe the relatively small literature that deals with multi-dimensional learning, in which new agents may combine the learning of various traits from different mentors.%
\footnote{
\ Our notion of multi-dimensional learning should not be confused with \citeauthor{arieli2019multidimensional}'s (\citeyear{arieli2019multidimensional}) different use of the same phrase. In their setup the phrase describes agents who take actions sequentially, and the order in which actions are taken is determined by a multi-dimensional integer lattice rather than a line as in the standard model of herding.} 
\cite{arad2012multi} and \cite{arad2018multi} present experimental evidence that perceive games with large set of strategies as multi-dimensional, and that they think how to play in each strategic dimension, independently of the other dimensions. 
Motivated by this, \cite{arad2019multidimensional} present a static (non-evolutionary) solution concept that is related to the case of $r=1$ under the assumption of each player perceiving the distribution of traits as uniform.

The combination of different traits in a new agent in this social learning situation resembles the way that genetic inheritance is passed through generations in sexual inheritance: namely, just as a newly formed individual inherits DNA that combines genes from both of her parents (in contrast to the standard replicator dynamic which resembles asexual inheritance), a new agent in our social learning model combines traits from several mentors. 
The stability of phenotypic behaviour that is determined by the combination of genes (at different loci in the DNA) has long been studied in the biological literature, see, e.g., \citealp{karlin1975general,eshel1984initial,matessi1996long}).

 \cite{waldman1994systematic} studies a setup in which the action of each agent is two-dimensional, where each dimension reflects  a finite choice regarding the level of  a different bias; for example, the first dimension may reflect the amount of overconfidence, and the second dimension the amount of disutility from work. \citeauthor{waldman1994systematic} shows that a pair of biases can be evolutionarily stable under sexual inheritance if the level of each bias is optimal when taking the level of the other bias as fixed (\citeauthor{waldman1994systematic} calls such pairs ``second-best adaptations”). 
 \citeauthor{frenkel2018endowment} (\citeyear{frenkel2018endowment}), extended this analysis to a setup in which the level of each bias is a continuum, and shows that in that case although second-best adaptations do not exist, biases, which approximately compensate for the errors that any one of them would give rise to in isolation, may persist for relatively long periods.%
 \footnote{\ Other related models that explain how pairs of biases, which approximately compensate for each other, can be stable are \cite{herold2023second,steiner2016perceiving,netzer2021endogenous}.}
 
In the last decade some research papers have studied the relation between evolutionary dynamics and learning 
algorithms (see, e.g., \citealp{ chastain2014algorithms, barton2014diverse, meir2015sex}). 
These papers showed that sexual inheritance achieves `regret minimisation' learning and convergence to Nash equilibria, and that it is helpful to regard genetic alleles in separate loci as playing a common interest game. 
\citet{edhan2017sex} points out an important advantage of sexual inheritance over asexual inheritance in setups in which the set of possible genetic combinations is large. In such setups, individuals in any generation can only bear a tiny sample of the large genotype space.
An asexual population samples once and finds a local maximum within that sample.
In contrast, by continuously re-sampling, the sexual population more reliably attains an asymptotically globally superior action.

\cite{palaiopanos2017multiplicative} studied the learning behaviour of the polynomial multiplicative weights update (MWU) algorithm. \citeauthor{palaiopanos2017multiplicative} showed that interior trajectories always converge to pure Nash equilibria in congestion games in which each player separately applies an MWU algorithm.
\citeauthor{edhan2021making} (\citeyear{edhan2021making}), noting that the replicator is a special case of the MWU algorithm and that genetic recombination reproduction can be cast as a potential game between genetic loci separately implementing a replicator dynamic, build on a similar result to show that haploid sexually reproducing populations exhibit monotonic increase in mean payoff and converge to pure Nash equilibria. 

A key difference between our proposed research and the existing literature is that the latter focuses on situations in which an agent's payoff essentially depends only on her own action, independently of the aggregate behaviour in the population. 
The strategic aspect of the payoff structure (i.e., the fact that an agent's payoff crucially depends on the behaviour of other agents in the population) is a key factor in our model. This dependency of the agent's payoff on the behaviour of others yields qualitatively different results. 

\section{Model}\label{sec-model}
\subsection{Basic Setup}
Let 
$G = (A,u)$ be a two-player symmetric 
normal form game, where $A$ is a (finite) set of actions and
$u:A^2 \to \R^{++}$ is a payoff function.
We interpret $u(a,a') \in \R^+$ as the payoff of an agent playing action $a$ against one playing
action $a'$.

A continuum of agents of mass one is presumed. Each agent is associated with an action $a \in A$; this identification is called the \emph{type} of the agent. 
The state space is the simplex $\Delta(A)$, where we interpret a population state (abbr., \emph{state})
$\a \in \Delta(A)$ as a distribution of types with $\a(a)$ 
expressing the frequency of agents in the population playing $a$.
The payoff (fitness) of an agent of type $a$ in state $\a$ is
\begin{equation}
	\label{eq:agentFitness}
	u_{\a}(a) \:= \sum_{a' \in A} \a(a') u(a,a'),
\end{equation}
and the average (mean) payoff in state $\a$ is 
\begin{equation}
	\label{eq:meanFitness}
	u_{\a} \:= \sum_{a \in A} \a(a) u_\a(a) = \sum_{a,a' \in A} \a(a) \a(a') u(a,a').
\end{equation}

We define a state $x$ to be a \emph{Nash equilibrium} if $u_{\a} \geq u_\a(a)$ for each $a \in A;$
i.e, in a Nash equilibrium no action can induce a payoff that is greater than the average payoff.

Moving from static descriptions to a dynamic law of motion, the time $t \ge 0$ is continuous. 
We denote an initial state of a 
trajectory by $\a^0$ and the state at time $t$ along this trajectory by $\a^t$.
The derivative with respect to time is denoted by
$\dot{\a} \:= \frac{d\a}{dt}$.

At each time $t$ there is a flow one of agents who die regardless of their type. Each dying agent is replaced by a new agent 
(or, equivalently, one can think of this flow as capturing agents who occasionally revise their strategies).

The following \emph{matrix notation} will be helpful.  Fix an arbitrary ordering of $A=(a_i)$, where we identify action $a_i$ with its index $i$. Let $U$ be the $|A|\times|A|$ payoff matrix. Given a subset of actions $\hat A\subseteq A$, let $U_{\hat A}$ be the payoff matrix restricted to the actions in  $\hat A$, and let $T_{\hat A}$ be the \emph{tangent space} of $\Delta(\hat A)$, i.e., the set of vectors $w \neq 0$ of size $|{\hat A}|$ such that $\sum_{a_i\in\hat A}(w_i)=0$. Let $w^\intercal$ denote the \emph{transpose} of the vector $w$.

\subsection{Strategic Dimensions}
If we were describing the standard imitation interpretation of the replicator dynamics (\citealt[Section 4.4.3]{weibull1997evolutionary}), we would at this point
imagine that a newly born agent selects a single `mentor' whose action is imitated. 
However, we introduce here an extension to the replicator dynamics, by
positing that there is a set $D = \{1, \ldots, |D|\}$ of dimensions of behaviours, with $|D| \ge 2$. 
A typical dimension will be denoted by $d$.
Each $d \in D$ is associated with a finite set $A_d$ of traits, and a typical trait in dimension $d$ will be denoted $a_d \in A_d$.
We also write $A_{-d} \:= \prod_{d' \ne d} A_{d'}$.

The set of actions $A$, which in most of the evolutionary game theory literature is just a collection of elements lacking internal features, is defined here as follows: each type $a \in A$ is now defined to be a $D$-tuple, i.e., $a = (a_d)_{d \in D}=(a_1,...,a_{|D|})$. 
 
The interpretation is that the set of traits within the set $A_d$ of dimension $d$ are mutually exclusive; an agent can exhibit only one trait $a_d \in A_d$.
In contrast, traits in different dimensions are complementary. 
With slight abuse of notation we write $a_d \in a$ when trait $a_d$ is one of the components in the $D$-tuple $a$.

We demonstrate the 
multi-dimensionality of types in the following example.
\begin{example}
[The Partially-Enforceable Prisoner's Dilemma]\label{ex:PD-contracts} Consider an interaction in which each player simultaneously makes two choices:
\begin {enumerate}
\item making either a  simple contract-enforceable promise (abbreviated, $\ss$) or an ambiguous non-enforceable promise (abbreviated, $\aa$) to cooperate, and
\item cooperating (abbr., $\cc$) or defecting (abbreviated, $\dd$)
in a prisoner's dilemma. 
\end {enumerate}
The set of actions is two-dimensional and includes $4=2^2$ actions: $A=(\sc,\sd,\ac,\ad)$, where the first dimension describes the type of promise, and the second dimension describes the behaviour in the prisoner's dilemma. An ambiguous promise induces a cost of 5. When both players cooperate, they obtain a payoff of 15. When a player's promise is ambiguous, a player can gain 6 by defecting 
(the net gain is 1, given the cost of 5 induced by the ambiguous promise), and in this case her opponent loses 9. By contrast, when a player's promise is simple, she loses 5 from defecting (without affecting her opponent's payoff.) The payoff matrix is summarised in Table \ref{tab:Payoff-Matrix-of}. Note that action $\ad$ strictly dominates action $\sc$, which, in turn, strictly dominates the two remaining actions $\sd$ and $\ac$.
\hfill \myenddefinition
\end{example}

\begin{table}[t]
\caption{\label{tab:Payoff-Matrix-of}Payoff Matrix of a Partially-Enforceable Prisoner's Dilemma}

\begin{centering}
\par\end{centering}
\centering{}%
\begin{tabular}{c|c|c|c|c|}
\multicolumn{1}{c}{} & \multicolumn{1}{c}{\textcolor{red}{$\sc$}} & \multicolumn{1}{c}{\textcolor{red}{$\sd$}} & \multicolumn{1}{c}{\textcolor{red}{$\ac$}} & \multicolumn{1}{c}{\textcolor{red}{$\ad$}}\tabularnewline
\cline{2-5} \cline{3-5} \cline{4-5} \cline{5-5} 
\textcolor{blue}{$\sc$} & \textcolor{blue}{$15$},~\textcolor{red}{15} & \textcolor{blue}{$15$},~\textcolor{red}{10} & \textcolor{blue}{$15$},~\textcolor{red}{10} & \textcolor{blue}{$6$},~\textcolor{red}{16}\tabularnewline
\cline{2-5} \cline{3-5} \cline{4-5} \cline{5-5} 
\textcolor{blue}{$\sd$} & \textcolor{blue}{$10$},~\textcolor{red}{15} & \textcolor{blue}{$10$},~\textcolor{red}{10} & \textcolor{blue}{$10$},~\textcolor{red}{10} & \textcolor{blue}{1},~\textcolor{red}{16}\tabularnewline
\cline{2-5} \cline{3-5} \cline{4-5} \cline{5-5} 
\textcolor{blue}{$\ac$} & \textcolor{blue}{$10$},~\textcolor{red}{15} & \textcolor{blue}{$10$},~\textcolor{red}{$10$} & \textcolor{blue}{10},~\textcolor{red}{10} & \textcolor{blue}{$1$},~\textcolor{red}{16}\tabularnewline
\cline{2-5} \cline{3-5} \cline{4-5} \cline{5-5} 
\textcolor{blue}{$\ad$} & \textcolor{blue}{$16$},~\textcolor{red}{6} & \textcolor{blue}{16},~\textcolor{red}{1} & \textcolor{blue}{$16$},~\textcolor{red}{1} & \textcolor{blue}{7},~\textcolor{red}{7}\tabularnewline
\cline{2-5} \cline{3-5} \cline{4-5} \cline{5-5} 
\end{tabular}
\end{table}
\subsection{Frequencies of Traits and Payoffs}

The payoff function $u: A \times A \to \R^+$ remains as before, as do the notations $u_{\a}(a)$ and $u_\a$ from Equations (\ref{eq:agentFitness}) and (\ref{eq:meanFitness}).
To this we add new expressions.
Let the (marginal) frequency of trait $a_d\in A_d$ 
(resp., trait profile $a_{-d}\in A_{-d}$) in  state $\a$ be denoted by
\[
\a(a_d) \:= \sum_{a_{-d} \in A_{-d}} \a(a_d, a_{-d}),\;\;\;\;\;\;\;\a(a_{-d}) \:= \sum_{a_{d} \in A_{d}} \a(a_d, a_{-d}).
\]
Let $\supp(\a)$ (resp., $\supp_d(\a)$,
$\supp_{-d}(\a)$) denote the set of actions (resp., traits, trait-profiles) with positive frequency in state $\a$, that is,
\begin{align}
	\label{eq:supp-def}
& \supp(\a) \:= \{a \in A \mid \a(a) > 0\},\,\,\,\, \supp_d(\a) \:= \{a_d \in A_d \mid \a(a_d) > 0\}\\
& \;\;\;\;\;\;\;\;\;\;\;\;\;\;\;\;\;\;\;\;\nonumber \supp_{-d}(\a) \:= \{a_{-d} \in A_{-d} \mid \a(a_{-d}) > 0\}.
\end{align}
Let $\INT(\D(A))$ denote the set of interior (full-support) states, i.e., $$\INT(\D(A))=\{(\a\in \D(A)|\supp(\a)=A\}.$$
\begin{definition}
   For a given state $\a$, we call the collection of actions such that each of the traits in each action has positive frequency in state $\a$ the \emph{rectangular closure of the support} of $\a$ and denote it by $\overline{\supp}(\a)$. In detail,
   \begin{equation}
	\label{eq:supp-rec-def}
        \overline{\supp}(\a) \:= \{a \in A \mid \a(a_d) > 0 \text{ for all }  a_d \in a\}.
    \end{equation}
    
    It is immediate that $\supp(\a) \subseteq \overline{\supp}(\a)$. We say that $\a$ has \emph{rectangular support} if $\supp(\a) = \overline{\supp}(\a)$. Let $\Rec(\D(A))$ denote the set of states with rectangular support. 
    Note that any interior state has rectangular support, i.e.,  $\INT(\D(A))\subseteq \Rec(\D(A))$.
    \hfill \myenddefinition
\end{definition}

These notions of support are illustrated by revisiting Example \ref{ex:PD-contracts}.
\addtocounter{example}{-1}
\begin{example}[continued] 
    Let $\a_1$ be the state that places weight $50\%$ on each of the types $\sc$ and $\ad$.
	Observe that state $\a_1$ does not have rectangular support: $\supp(\a_1) = \{\sc, \ad\}\neq\overline{\supp}(\a_1) = A$. Let $\a_2$ be the state that places weight $50\%$ on each of the types $\sc$ and $\sd$. Then state $\a_2$ has rectangular support: $\supp(\a_2) =\overline{\supp}(\a_2) = \{\sc, \sd\}$. \hfill \myenddefinition
\end{example}

For each $a_d \in \supp_d(\a)$ define $u_\a(a_d)$ to be the 
mean (marginal) payoff of agents with trait $a_d$:
\begin{equation}
	\label{eq:trait-payoff}
u_\a(a_d) \:= \frac{1}{\a(a_d)}\sum_{a_{-d} \in A_{-d}} \a(a_d, a_{-d}) \cdot u_\a (a_d, a_{-d}).
\end{equation}

The definitions of the trait frequencies and payoffs are illustrated as follows.
\addtocounter{example}{-1}
\begin{example}[continued] 
Let $\a$ place weight $40\%$ on type $\sc$, $30\%$ on $\sd$,  $20\%$ on $\ac$, and $10\%$ on $\ad$. These weights imply that the marginal frequencies are: $\a(\ss)=40\%+30\%=70\%$, $\a(\aa)=1-\a(\ss)=30\%$, $\a(c)=60\%$, $\a(\dd)=40\%$. A simple calculation shows that the types' payoffs are: $u_{\a}(\sc)\approx14.1$,  $u_{\a}(\sd)=u_{\a}(\ac)\approx9.1$, $u_{\a}(\sd)\approx15.1,$  
and that the mean payoff in the population is $u_{\a}\approx11.7$.
Applying (\ref{eq:trait-payoff}) implies that trait payoffs are: $u_{\a}(\ss)=\frac{40\% \cdot 14.1 + 30\% \cdot 9.1}{70\%}\approx12$, $u_{\a}(\aa)=\frac{20\% \cdot 9.1 + 10\% \cdot 15.1}{30\%}\approx11.1$, $u_{\a}(\cc)=12.4$, $u_{\a}(\dd)=10.6$. 
\hfill \myenddefinition
\end{example}

Next, we observe that the average marginal payoff in each dimension is equal to the average payoff of the population $u_\a$. This is so because:
\begin{equation}
	\label{eq:average_marginal_payoff}	
		\sum_{a_d \in A_d} \a(a_d)u_{\a}(a_d) =  \sum_{a_d \in A_d} \sum_{a_{-d} \in A_{-d}} \a(a_d, a_{-d}) u_\a (a_d, a_{-d}) 
		= \sum_{a \in A} \a(a) u_{\a}(a)
		= u_\a.\qedhere
		\end{equation}	

\subsection {Recombinator Dynamics}
In our model, 
a new agent may either with probability $1-r$ select a single 
incumbent (mentor) and directly imitate all the traits of that mentor, or with probability $r$ sample $|D|$ mentors, one for each dimension, and imitate a trait from each of those mentors. In greater detail, if the new agent has a single mentor, then that mentor is sampled from the current state $\a$. 
If instead there are multiple mentors, the new agent independently samples $|D|$ different mentors and for each $d$ imitates the $d$-th trait of the $d$-th mentor. Importantly, the sampling of mentors is not uniform: the more successful a mentor of type $a$ is at 
state $\a$, as measured by that mentor's payoff relative to the population average, the more mentees he or she attracts.
At the same time, the greater $\a(a)$, that is, the greater the proportion of action $a$ in population $\a$, the greater the likelihood that one of the mentors of type $a$ will be selected.

Calling $r \in [0,1]$ the 
\emph{recombination rate}, the resulting \emph{recombinator dynamics} is 
\begin{equation}
	\label{eq:recombinatorMotion}
	\dot{\a}(a) = (1-r)\frac{\a(a)u_\a(a)}{u_\a} + r\prod_{a_d \in a} \frac{\a(a_d)u_\a(a_d)}{u_\a} - \a(a),  
\end{equation}
where the first component is the inflow of new agents who imitate a single mentor, the second component is the inflow of new agents who combine learning from $|D|$ mentors, and the last term is the outflow of dying agents.

The initial state of a trajectory is denoted by $\a^0$, and the dynamics of a trajectory is determined by the equation of motion given by Equation (\ref{eq:recombinatorMotion}), with $\a^t$ substituting for $\a$ and $\dot{\a^t} \equiv \frac{d\a^t}{dt}$ denoting the time derivative of the state. 

When $r=0$, Equation (\ref{eq:recombinatorMotion}) reduces to:
\begin{equation}
	\label{eq:replicatorMotion}
	\dot{\a}(a) = \frac{\a(a)u_\a(a)}{u_\a} - \a(a)= \frac{1}{u_\a}\a(a)(u_\a(a) - u_\a),
\end{equation} 
which is the replicator dynamics (up to a payoff-dependent rescaling of time; see \citealp[Section 4.4.3]{weibull1997evolutionary}).

When $r=1$, we have a model of pure combination of traits at each time by each agent that we call the 
\emph{combinator dynamics}, characterised by the equation of motion:
\begin{equation}
	\label{eq:combinatorMotion}
	\dot{\a}(a) = \prod_{a_d \in a} \frac{\a(a_d)u_\a(a_d)}{u_\a} - \a(a).
\end{equation}

\begin{remark}
\citet[Section 4.4.3]{weibull1997evolutionary} presents a more general imitation dynamic (for the case of $r=0$) in which $u_\a(a)$ (resp., $u_\a$) in Equation (\ref{eq:replicatorMotion}) is replaced by $w(u_\a(a))$ (resp., $w(u_\a)$), where $w:\mathbb{R}\rightarrow\mathbb{R}^{++}$ is a strictly monotone function. Our dynamics can capture this general version by a normalisation of the payoff function. That is, if the original payoff function is denoted $\pi:A\rightarrow\mathbb{R}$ (which might be measured in dollars), then $u \equiv w(\pi)$ is the normalised payoff following a monotone transformation to cardinal units, measuring the probability
of being chosen as a mentor that is induced by the dollar payoff. \hfill \myenddefinition
\end{remark}

\subsection{Forward Invariance}
It is well known that under the replicator dynamics the support of any state remains identical along trajectories at all finite times $t \ge 0$; this property is called \emph{forward invariance}.
The support may decrease (but not increase) as $t \to \infty$ (i.e., if $r=0$, then $\supp(\a^t) = \supp(\a^0)$  $\forall t > 0,$ and
$\lim_{t \to \infty} \supp(\a^t)$ $\subseteq \supp(\a^0)$).
Moreover, these properties hold in the broader class of imitative dynamics 
(as defined in \citealt[Section 5.4]{sandholm2010population}). 

A related property holds for the recombinator dynamics with $r>0$ with one key difference: the support, if it is not rectangular at time zero, instantaneously increases to its rectangular closure for any $t>0$ (as demonstrated in the example below). That is, for any $r > 0$ and any trajectory starting at $\a^0$:
\begin{equation}
	\label{eq:supports}
	\supp(\a^t) = \overline{\supp}(\a^0)\textrm{ for all } t > 0, \textrm{ \,\,\,\,and \,\,\,\,} \lim_{t \to \infty} \supp(\a^t) \subseteq \overline{\supp}(\a^0).
\end{equation}
\addtocounter{example}{-1}
\begin{example}[continued] 
Let the initial state $\a^0$ place positive weights on types $\sc$ and $\ad$ (i.e., $\supp(\a^0) = \{\sc,\ad\})$. Observe that $\supp(\a^t) = \overline{\supp}(\a^0) = A$ for all $t > 0$. This is because every time a new agent is born there is a positive probability that a mentor pair $(\sc,\ad)$ will be sampled, leading to the creation of type $\sd$, and similarly a positive probability that mentor pair $(\ad,\sc)$ will be sampled (the ordering makes a difference), leading to the creation of type $\ac$. Note that this always holds, even though actions $\sd$ and $\ac$ induce strictly dominated payoffs.
\hfill \myenddefinition
\end{example}

The recombinator dynamics exhibit the same continuous (and forward invariant) behaviour as the imitative dynamics (\citealt[Section 5.4]{sandholm2010population}) at all positive times $t>0$. 
Instantaneous discontinuities in the recombinator dynamics can only occur at time zero, and only when the support of the initial state is not rectangular. 

\subsection {Stability and Convergence}
We conclude this section with a few standard definitions of dynamic stability and convergence. 
A state is stationary if it is a fixed point of the dynamics.
\begin{definition}
   A state $\a \in \D(A)$ is \emph{stationary} if $\dot\a(a)=0$ for each $a \in A$.
\end{definition}
It is well-known that any convergent limit state of a trajectory $\lim_{t \to \infty} \a^t$ must be a stationary state (Proposition 6.3 of \citealp{weibull1997evolutionary}).

A state is Lyapunov stable if the trajectory of a population starting out near that state
always remains close, and it is asymptotically stable if, in addition, the trajectory eventually converges to the equilibrium state. A state is
unstable if it is not Lyapunov stable. It is well known (see, e.g., \citealp[Section 6.4]{weibull1997evolutionary})
that every Lyapunov stable state is stationary.

\begin{definition}
\label{def:lyaponouv-stability}
A state $\a\in\D (A)$
is \emph{Lyapunov stable} if for every neighbourhood $U$ of $\a^*$
there is a neighbourhood $V\subseteq U$ of $\a^*$ such that if the
initial state of a trajectory satisfies $\a_0\in V$ then $x^t\in U$
for all $t>0$. A state is \emph{unstable} if it is not
Lyapunov stable.
\end{definition}

A Lyapunov stable state is asymptotically stable if starting from every sufficiently close initial condition, the trajectories converge to that state.

\begin{definition}\label{def:asymptotic-stable}
A state $\a^{*}\in\D(A)$ is \emph{asymptotically stable} if (1) it is Lyapunov stable and (2) there is an open neighbourhood $U$
of $\a^{*}$ such that all trajectories initially in $U$ converge
to $\a^{*},$ i.e., $\a^0\in U$ $\Rightarrow$ $\lim_{t\rightarrow\infty}\a^t=\a^{*}$.  
\end{definition}

The \emph{basin of attraction} of a state $x^*$, which is denoted by $BA(\a^*)$ is the set of initial states that converge to this state, i.e.,
\begin{equation}\label{eq:basin-of-attraction-def}
    BA(\a^*)=\{\a \in A | \a^0=\a \Rightarrow \lim_{t\rightarrow\infty}\a^t=\a^{*}.\}
\end{equation}
Finally, we say that state $\a^*$ is \emph{globally stable} if its basin of attraction includes all interior states, i.e., $\INT(\D(A))\subseteq BA(\a^*).$

\section{Payoff Monotonicity and Stationary States} \label{sec-staionary}
In this section we explore the monotonicity properties of the recombinator dynamics. We first demonstrate that the recombinator dynamics violates payoff monotonicity. We then define a new payoff function, called $r$-payoff, and we show that the recombinator dynamics is monotone with respect to the  $r$-payoffs. Finally, we show that the induced dynamics on the traits (rather than on the types) does satisfy payoff monotonicity (with respect to the original payoff function $u$).

\subsection {Non-Monotonicity of the Recombinator Dynamics}
 Dynamics are payoff monotone (see, e.g., \citealt[Definition 4.2]{weibull1997evolutionary}) if a type with a higher payoff grows at a higher rate.
 \begin{definition}
    Dynamics $\dot\a$ are \emph{payoff monotone} if $u_\a (a)>u_\a (a') \iff \frac{\dot\a(a)}{\a(a)}>\frac{\dot\a(a')}{\a(a')}$ for each state $\a \in \Delta(A)$ and for each pair of types $a,a' \in \supp(\a)$.
 \end{definition}
  It is well known that the replicator dynamic satisfies payoff monotonicity, which implies that its stationary states are those that satisfy the property that all incumbent types have the same payoff. Formally 
 \begin{fact}(\citealp[Proposition 5.9]{weibull1997evolutionary})~\label{claim-replicator-stationary}
 \begin {enumerate}
 \item The replicator dynamic is payoff monotone.
 \item A state $\a$ is stationary if and only if $u_\a(a)=u_\a(a')$ for all $a,a'\in\supp(\a).$
 \end{enumerate}
 \end{fact}
 
 Given this, it is noteworthy that our next example demonstrates that the recombinator dynamics violate payoff monotonicity when the recombination rate is positive, and that this can allow strictly dominated types to be asymptotically stable. 
  \addtocounter{example}{-1}
\begin{example}[continued] 
Fix 
  a sufficiently small $\e<<1$. Consider an initial state $\a$ that puts weight $1-\e$ on type $\sc$ and weight $\e$ on type $\ad$. The payoff matrix  (Table \ref{tab:Payoff-Matrix-of}) implies that $u_\a(\sc)=15-9\e<u_\a(\ad)=16-9\e$. In what follows, we show that $\frac{\dot\a(\sc)}{\a(\sc)}>\frac{\dot\a(\ad)}{\a(\ad)}$, which violates payoff monotonicity. Observe that $u_\a(\ss)=u_\a(\cc)=15-9\e$, $u_\a(\aa)=u_\a(\dd)=16-9\e$, and $u_\a\approx15-8\e$. Substituting these values in the recombinator dynamics (Equation \ref{eq:recombinatorMotion}) yields: 
 \[
\dot\a(\sc)=(1-r)\frac{(1-\e)(15-9\e)}{15-8\e}+r\left(\frac{(1-\e)(15-9\e)}{15-8\e}\right)^2-(1-\e)=O(\e)\Rightarrow \frac{\dot\a(\sc)}{\a(\sc)}=O(\e),
\]
\[
\dot\a(\ad)=(1-r)\frac{\e(16-9\e)}{15-8\e}+r\left(\frac{\e(16-9\e)}{15-8\e}\right)^2-\e=\e\frac{16(1-r)}{15}-\e+O(\e^2)\Rightarrow \frac{\dot\a(\ad)}{\a(\ad)}=\frac{1-16r}{15}+O(\e).
\]
 Observe that for a sufficiently small $\e$, $\frac{\dot\a(\sc)}{\a(\sc)}>\frac{\dot\a(\ad)}{\a(\ad)}$ if $r>\frac{1}{16}$, while the opposite inequality holds if $r<\frac{1}{16}$. 
We later show (Corollaries \ref{cor-pure-nec}--\ref{cor-pure-suf} ) that the strictly dominated type $\sc$ is an asymptotically stable state if $r>\frac{1}{16}$ and it is unstable if $r<\frac{1}{16}$, and that the state $\ad$ is asymptotically stable for all values of $r$ (and one can show that are no other asymptotically stable states). Figure \ref{fig:PDWC-near-ac-r-half} illustrates the growth rates, phase plots and the basins of attractions of $\ac$ and $\sd$ for two values of $r$: 0.1 and 0.9.
 \end{example}

\begin{figure}
\begin{center}
\caption{Relative Growth Rates and Phase Portraits in Example
 \ref{ex:PD-contracts}.}
 \label{fig:PDWC-near-ac-r-half}
 \includegraphics[scale=0.385]{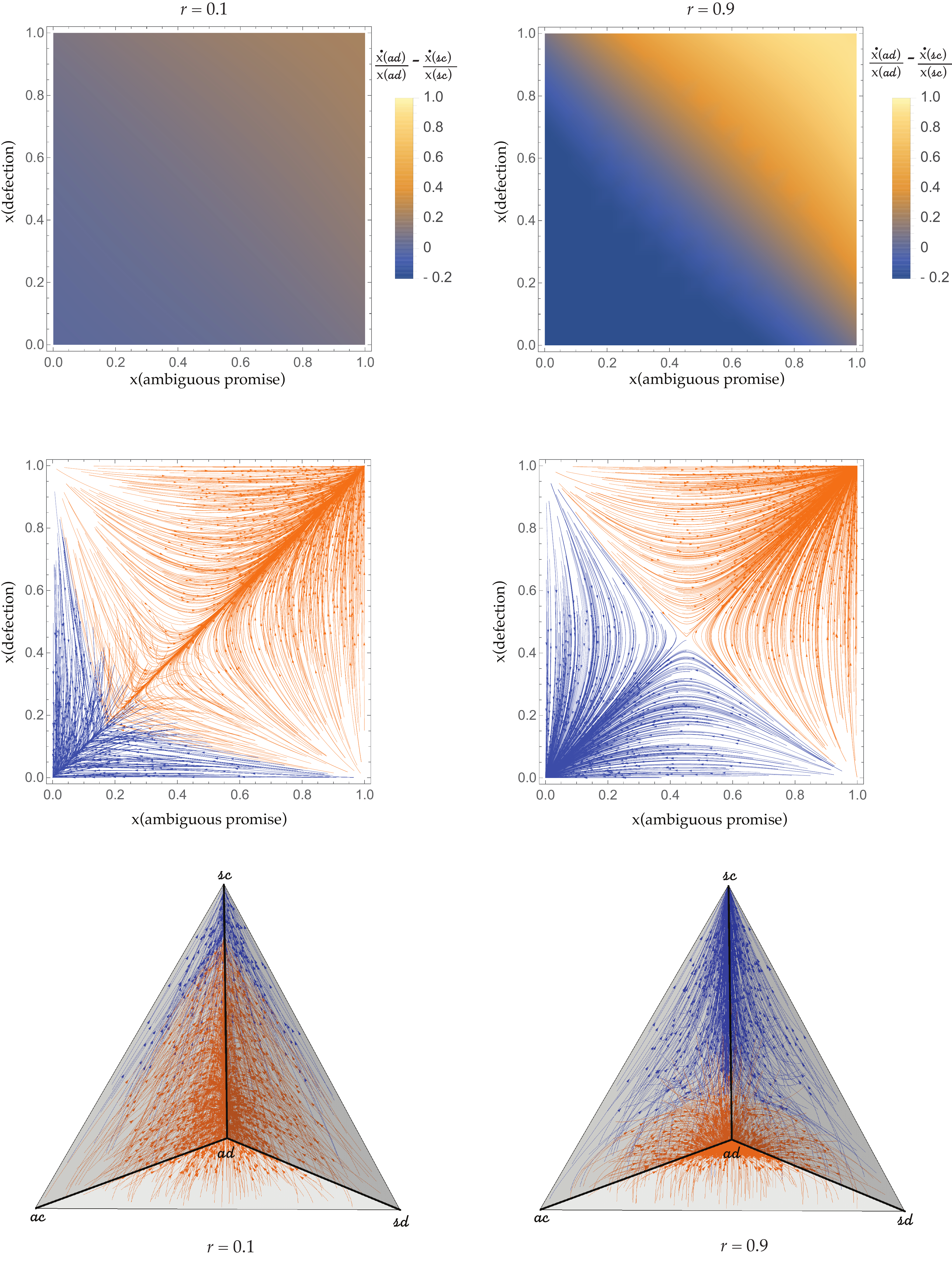}
  \end{center}
   {\footnotesize{}\
    {\small{}
   {This figure illustrates the relative growth rates and the basins of attraction in the partially-enforceable Prisoner's Dilemma. In all the panels, the left side describes recombination rate of $r=0.1$ and the right side describes $r=0.9$. The upper panel illustrates  the difference between the relative growth rates of the types $\ad$ and $\sc$ (i.e., $\frac{\dot\a(\ad)}{\a(\ad)}-\frac{\dot\a(\sc)}{\a(\sc)}$) under the recombinator dynamics given trait-independent states. 
   The $x \, (y)$ axis describes the frequency of  trait $\aa$ ($\dd$).  
   The middle panel illustrates the evolution of the projection of the trajectories on the plane defined by the frequencies of the traits $\aa$ and $\dd$. The bottom panel illustrates the evolution of the trajectories (i.e., the phase plot) in the full three dimensional space of $\Delta(A)$.} Trajectories that converge to $\sc$ (everyone giving a simple promise and cooperating) appear in blue, and those that converge to $\ad$ (everyone giving an ambiguous promise and defecting) appear in orange. 
   }
   }{\footnotesize\par}
 \end{figure}

 \subsection {$r$-Payoffs}
In this section, we define a new payoff function $z^r_\a(a)$, and show that the recombinator dynamics is monotone with respect to $z^r_\a(a)$. 

The recombinator dynamics of Equation (\ref{eq:recombinatorMotion}) can be rewritten as:
 \begin{equation}
	\label{eq:recombDynamics2}
	\dot{\a}(a) = \a(a)(1-r) \frac{ u_\a(a)}{u_\a} +  r \prod_{a_d \in a} \a(a_d) \prod_{a_d \in a} \frac{u_\a(a_d)}{u_\a} - \a(a).
\end{equation}
This implies that for any $a \in \supp(\a)$:
\begin{align}
	\label{eq:relativeGrowth1}
	\frac{\dot{\a}(a)}{\a(a)} &= (1-r) \frac{ u_\a(a)}{u_\a} +    r  \frac {\prod_{a_d \in a}\a(a_d)}{\a(a)}    \prod_{a_d \in a} \frac{u_\a(a_d)}{u_\a} - 1.  
\end{align}
Now define for any $a \in \supp(\a)$ the \emph{trait-to-type ratio} $m_\a(a)$
\begin{equation}
\label{eq:traitToType}
m_\a(a) \:= \frac {\prod_{a_d \in a}\a(a_d)}{\a(a)}.
\end{equation}
The trait-to-type ratio $m_\a(a)$ is the ratio between the product of the weights of the traits in $\a$ to the weight of $\a$ itself. In states in which the event of a randomly chosen agent having a trait in one dimension (say, being cooperative in Example \ref{ex:PD-contracts}) is independent of that agent having a trait in another dimension (say, giving a simple promise), $m_\a(a)\equiv 1$. We call such states trait independent.

\begin{definition}
   A state $\a$ is \emph{trait independent} if $\a(a)=\prod_{a_d \in a}x(a_d)$ for each $a\in A$.
\end{definition}
Thus, the trait-to-type ratio $m_\a(a)$ captures the distance of $\a(a)$ from trait independence. Values of $m_\a(a)<1$ (respectively, $m_\a(a)>1$) represent positive correlation between the traits composing type $a$, in the sense that that the probability of a randomly chosen agent having type $a$ is larger (respectively, smaller) than the product of the probabilities  of $|D|$ randomly chosen agents each having one of the traits $a_d$ in $a$.
Define for each state $\a$ and each action $a\in\supp(\a)$:
\begin{equation}\label{eq:zed_def}
	z^r_\a(a) \:= (1-r) \frac{ u_\a(a)}{u_\a} +  r m_\a(a) \prod_{a_d \in a} \frac{u_\a(a_d)}{u_\a}.
\end{equation}
We call $z^r_{\a}(a)$ the \emph{$r$-payoff} of action $a\in\supp(\a)$ at state $\a$. 

Observe that without recombination ($r=0$) the $r$-payoff coincides with the standard payoff function $u_\a(a)$  (up to a normalisation attained by dividing by the mean payoff $u_\a$):
$z^0_{\a}(a)=\frac{ u_\a(a)}{u_\a}$. 
Further observe that in the opposite case of full recombination ($r=1$) the $r$-payoff of a strategy $a$ depends only on the payoffs of its traits (and not of its own payoff $u_\a(a)$): $z^1_{\a}(a)=m_\a(a) \prod_{a_d \in a} \frac{u_\a(a_d)}{u_\a}$. 
In the general case of $r\in(0,1)$, the $r$-payoff is a convex combination of these two terms: $z^r_{\a}(a)=(1-r)z^0_{\a}(a)+rz^1_{\a}(a)$.

Substituting $z^0_{\a}$ and $m_\a(a)$ in Equation (\ref{eq:relativeGrowth1}) yields:
\begin{equation}
	\label{eq:recombDynamicsSimple}
	\frac{\dot{\a}(a)}{\a(a)} = z^r_\a (a) - 1.  
\end{equation}
The crucial point here is that re-casting the recombinator dynamic in the form of Equation (\ref{eq:recombDynamicsSimple}) shows that, when interpreted in terms of the vector field determined by $z^r$ instead of  the $u$ payoff vector field, the recombinator exhibits replicator-like behaviour. 
In particular, Equation (\ref{eq:recombDynamicsSimple}) implies that the recombinator dynamic is monotone with respect to 
$z^r_\a (a)$, and that a state $\a$ is stationary if and only if  all incumbent types have $r$-payoff of 1. 
This is formalised in Proposition \ref{pro:stationary-states-recombinator}, which generalises Fact \ref{claim-replicator-stationary} to recombinator dynamics, where the game payoff function $u$ is replaced by the $r$-payoff function $z^r$.

\begin{proposition}
\label {pro:stationary-states-recombinator}
\label{prop-1}
The recombinator dynamics 
satisfy:
\begin{enumerate}
    \item $r$-payoff monotonicity:  $z^r_\a (a) > z^r_\a (a') \Leftrightarrow \frac{\dot\a(a)}{\a(a)} > \frac{\dot\a(a')}{\a(a')}$ for each state $\a \in \Delta(A)$ and for each pair of types $a,a' \in \supp(\a)$.
    \item A state $\a$ is stationary if and only if $z^r_\a(a) = 1$ for all $a \in \supp(\a).$
 \end{enumerate}
\end{proposition}
\begin{proof}
	Part (1) is immediate from Equation (\ref{eq:recombDynamicsSimple}). We prove part (2). Suppose that $z^r_\a (a) = 1$ for all $a \in \supp(\a)$.
	Then $\dot{\a}(a) = 0$ for $a \in \supp(\a)$, by Equation (\ref{eq:recombDynamicsSimple}).
	Hence each $a \in \supp(\a)$ satisfies the property that $\a^t(a) = \a^0(a)\equiv x(a)$ for all $t \ge 0$, 
 which implies that $x$ is stationary. For the other direction of part (2), suppose that $z^r_\a(a) \ne 1$ for some $a \in \supp(\a)$. 
	Then by Equation (\ref{eq:recombDynamicsSimple}), $\dot{\a}(a) \ne 0$ and hence stationarity cannot obtain.
\end{proof}

\subsection {Payoff Monotonicity of the Trait-Centric Dynamics}
The recombinator dynamics, which is defined over the set of types, induces dynamics over the set of traits. These induced dynamics can be interpreted as a game between the traits (which lies behind the original game between the types $G$). 
The gene-centered view of genetic evolution (\citealp{Williams-George-1966,r.dawkins1976the-selfish-gen}) highlights the ways in which  biological natural selection chooses  fitness-maximising genes, rather than choosing fitness-maximising individuals. 
Similarly, in what follows, we show that the social learning process that is captured by the recombinator dynamics leads to the survival of  payoff-maximising traits (rather than payoff-maximising types), where each trait $a_d \in A_d$ is essentially competing against the other traits in $A_d$.

Fix any trait $a_d$ in the support of $\a$ (i.e., $\a(a_d) > 0$). 
Let us slightly rewrite Equation (\ref{eq:recombDynamics2}) from the perspective of a particular $a_d \in a$:
\begin{equation}\label{eq:recombDynamics-trait-initial}
    \dot{\a}(a_d,a_{-d}) = (1-r) \frac{\a(a_d,a_{-d}) u_\a(a_d,a_{-d})}{u_\a} + r\frac{\a(a_d)u_\a(a_d)}{u_\a^{|D|}}
	\prod_{a_{d'} \in a_{-d}} \a(a_{d'})u_\a(a_{d'}) 
	- \a(a_d,a_{-d}). 
\end{equation}

Dividing Equation (\ref{eq:recombDynamics-trait-initial}) by $\a(a_d)$ and summing over all $a'_{-d} \in A_{-d}$ yields the following \emph{trait-centric recombinator dynamics}:
\begin{align}
	\label{eq:recombDynamics-trait}
	\frac{\dot{\a}(a_d)}{\a(a_d)} =\sum_{a'_{-d} \in A_{-d}} \frac {\dot{\a}(a_d, a'_{-d})}{x(a_d)} &= \frac{(1-r)} {u_x} u_\a(a_d)   
	+ \frac{r} {u_x^{|D|}} u_\a(a_d)  \sum_{a'_{-d} \in A_{-d}} \prod_{a'_{d'} \in a'_{-d}} \a(a'_{d'})u_\a(a'_{d'}) - 1 \\
	& = u_\a(a_d) \biggl(\frac{(1-r)} {u_x}   
	+ \frac{r} {u_x^{|D|}}  \sum_{a'_{-d} \in A_{-d}} \prod_{a'_{d'} \in a'_{-d}} \a(a'_{d'})u_\a(a'_{d'}) \biggr) - 1. \nonumber
\end{align}

Note that the right-hand side of Equation (\ref{eq:recombDynamics-trait}) can be decomposed into an expression involving $u_\a (a_d)$ and a sum involving only elements of $a_{-d}$.
This implies that the trait-centric recombinator dynamics is monotone in trait payoffs. 
This yields a simple trait-centred characterisation of stationary states: (a) all traits obtain the same payoff, and (b) all types obtain the same $r$-weighted average of the \emph{relative type payoff} (i.e., the ratio $\frac{ u_\a(a)}{u_\a}$ between the type's payoff and the average payoff ) and the trait-to-type ratio, $m_\a(a)$. 

\begin{proposition} 
\label{prop-trait-stationary} \label{prop-2}
The 
recombinator dynamics 
satisfy:
	\label{prop:relative}
	\begin{enumerate}
    \item Trait payoff monotonicity:  $u_\a (a_d) > u_\a (a'_d) \Leftrightarrow \frac{\dot\a(a_d)}{\a(a_d)} > \frac{\dot\a(a'_d)}{\a(a'_d)}$ for each state $\a \in \Delta(A)$ and for each pair of traits $a_d,a'_d \in \supp_d(\a)$.
    \item A state $\a$ is stationary if and only if
   \begin {enumerate}
    \item  \label {trait-equal-payoff} 
    $u_{x}(a_{d})=u_{x}$ for any dimension $d\in D$ and any $a_{d}\in supp_{d}(x)$, and 
	\item $(1-r) \frac{ u_\a(a)}{u_\a} +  r m_\a(a) = 1$ for any $a\in\supp(\a)$. 	
	\end{enumerate}
 \end{enumerate}
\end{proposition}

\begin{proof}~
\begin{enumerate}
    \item Part (1) is implied by the fact that replacing $a_d$ with $a'_d$ in the right-hand side of Equation (\ref{eq:recombDynamics-trait}) leaves the expression between the round brackets unchanged.
Appealing to Equation (\ref{eq:average_marginal_payoff}), we further conclude that $u_{x^*}(a_{d})=u_{x^*}(a'_{d})=u_{x^*}$.
\item 
Observe that (a) and (b) jointly imply that $z^r_\a(a) = 1$ for all $a \in \supp(\a)$, which implies by Proposition \ref{pro:stationary-states-recombinator} that $\a$ is stationary. 
For the other direction, suppose that  $\a$ is stationary. 
Proposition \ref{pro:stationary-states-recombinator} implies that  $z^r_\a(a) = 1$ for all $a \in \supp(\a)$. 
By Part (1), the stationarity of $\a$ implies that $u_{x}(a_{d})=u_{x}$ for every $d$ and $a_d\in\supp_d(\a)$. 
Substituting this equality in the definition of  $z^r_\a(a)$ 
(Equation (\ref{eq:zed_def})) implies that $(1-r) \frac{ u_\a(a)}{u_\a} +  r m_\a(a)=1$ for any $a\in\supp(\a)$. \qedhere 
\end{enumerate}
\end{proof}

Part (1) of Proposition \ref{prop-trait-stationary} implies that the recombinator dynamics selects for traits with higher 
(state-dependent) payoffs, and that traits that
consistently have lower payoffs become extinct. 
When the population converges to a stationary state, it must be the case that all surviving traits have exactly the same payoff.
In contrast, the surviving types in a stationary state may have different payoffs, under the constraint that types with higher payoffs have lower trait-to-type ratios.  This is illustrated in the following example.
\addtocounter{example}{-1}
\begin{example}[continued]
Consider the partially-enforceable prisoner's dilemma with recombination rate $r=0.5$. Consider the state $x^*$ 
(as described in Table \ref{tab:Heterogeneous-Stationary-State}) in which 38.4\% of the agents have type $\sc$, 23.9\% have type $\ad$, and 18.8\% have each of the remaining types ($\ac$ and $\sd$). The payoff matrix (Table \ref{tab:Payoff-Matrix-of}) implies that $\ad$ has the highest payoff of 13.8, $\sc$ has a payoff of 12.8, and the remaining two types have the lowest payoff of 7.8. Calculating the payoff of each trait (as the weighted mean of the payoffs of the types that has this trait) shows that each trait has the same payoff of 11.2 (satisfying condition (2-a) on Proposition \ref{prop-trait-stationary}). Observe that the average of each type's relative payoff $\frac{u_\a(\aa)}{u_\a}$ and its trait-to-type ratio $m_\a(\aa)$ is equal to one, which satisfies condition (2-b). For example, the highest relative payoff of 1.24 of $\ad$ is compensated by having the lowest trait-to-type ratio of 0.76. Thus, Proposition \ref{prop-trait-stationary} implies that $\a^*$ is stationary.
\end {example}
\begin{center}
\begin{table}[t]

\caption{\label{tab:Heterogeneous-Stationary-State} Stationary
State $x^{*}$ in the Partially-Enforceable Prisoner's Dilemma}

\begin{doublespace}
\begin{centering}
\begin{tabular}{|c|c|c|c|c|}
\hline 
\multicolumn{5}{|c|}{Types at the stationary state $x^{*}$}\tabularnewline
\hline 
$a$ & $x\left(a\right)$ & $u_{x}\left(a\right)$ & $\frac{u_{x}\left(a\right)}{u_{x}}$ & $
m_{x}\left(a\right)
$\tabularnewline
\hline 
$\sc$ & 38.4\% & 12.8 & 1.15 & 0.85\tabularnewline
\hline 
$\ac$ & 18.8\% & 7.8 & 0.70 & 1.3\tabularnewline
\hline 
$\sd$ & 18.8\% & 7.8 & 0.70 & 1.3\tabularnewline
\hline 
$\ad$ & 23.9\% & 13.8 & 1.24 & 0.76\tabularnewline
\hline 
\end{tabular}~~~~~%
\begin{tabular}{|c|c|c|}
\hline 
\multicolumn{3}{|c|}{Traits at the stationary state $x^{*}$}\tabularnewline
\hline 
$a_{d}$ & $x\left(a_{d}\right)$ & $u_{x}\left(a_{d}\right)$\tabularnewline
\hline 
$\cc$ & 38.4\%+18.8\%=57.2\% & $\frac{38.4\%}{57.2\%}$$\cdot$12.8+$\frac{18.8\%}{57.2\%}$$\cdot$7.8=11.2\tabularnewline
\hline 
$\dd$ & 23.9\%+18.8\%=42.7\% & $\frac{23.9\%}{42.7\%}$$\cdot$13.8+$\frac{18.8\%}{42.7\%}$$\cdot$7.8=11.2\tabularnewline
\hline 
$\ss$ & 38.4\%+18.8\%=57.2\% & $\frac{38.4\%}{57.2\%}$$\cdot$12.8+$\frac{18.8\%}{57.2\%}$$\cdot$7.8=11.2\tabularnewline
\hline 
$\aa$ & 23.9\%+18.8\%=42.7\% & $\frac{23.9\%}{42.7\%}$$\cdot$13.8+$\frac{18.8\%}{42.7\%}$$\cdot$7.8=11.2\tabularnewline
\hline 
\end{tabular}
\par\end{centering}
\end{doublespace}
The left table describes the frequencies, payoffs and trait-to-type
ratios of the types in  state $x^{*}$. The right table describes
the frequencies and payoffs of the traits in $x^{*}$.

\end{table}
\par\end{center}

An interesting observation arises when considering part 2 of Proposition \ref{prop-trait-stationary} in the special case of $r = 1$.
In that case  $(1-r) \frac{ u_\a(a)}{u_\a} +  r m_\a(a)=1$ reduces to $m_\a(a) = 1$.
In other words, under the combinator dynamic (i.e., $r=1$) a stationary state $x$
must be trait Independent (i.e., the combinator dynamic induces an exact trait-to-type ratio  of 1 at each $a \in \supp(x)$).

\section {Characterisation of Stable States}\label{sec-stable}
\subsection{Benchmark Result for the Replicator Dynamics}
It is well known that asymptotic stability of stationary states under the replicator dynamic is  characterised by two conditions:
\begin{enumerate}
    \item \emph{Internal stability}: the payoff matrix restricted to the incumbent types is (semi-) negative-definite
     with respect to the tangent space.
   
     \item \emph{External stability}: The payoffs of types outside the support is lower than the incumbents' payoff.
\end{enumerate} 
The strict variants of these conditions  (i.e., negative-definiteness and quasi-strictness) imply asymptotic stability, and their weak counterparts (semi-negative-definiteness and being a Nash equilibrium) are implied by asymptotic stability. 
This is formalised as follows. 

\begin{fact}\label{claim-necc-asymp-replic}
If state $x$ is asymptotically stable under the replicator dynamics, then it satisfies
\begin{enumerate}
    \item weak internal stability: 
    $w^\intercal\cdot U_{\supp(x)}\cdot w\leq0$ for each $w\in T_{\supp(x)}$; and
    \item weak external stability: $U_\a \geq U_\a(a)$ for each $a \notin \supp(\a)$.
\end{enumerate}
\end{fact}

\begin{fact}\label{claim-suffic-asymp-replic}
Stationary state $x$ is asymptotically stable under the replicator dynamics if it satisfies:
\begin{enumerate}
    \item internal stability:     $w^\intercal\cdot U_{\supp(x)}\cdot w<0$ for each $w\in T_{\supp(x)}$; and
    \item external stability: $U_\a>U_\a(a)$ for each $a \notin \supp(\a)$.
\end{enumerate}
\end{fact}
We omit the proofs of these well-known results (which are implied by combining Theorem 9.2.7, Corollary 9.4.2, Theorem 9.4.4, and Theorem 9.4.8 in \citealp{van1991stability}). 

The main results of this section generalise these facts, and characterise asymptotic stability under the recombinator dynamics. In order to do so, we first need to present the notion of $r$-Jacobian matrix (which generalises the payoff matrix in condition (1) of internal stability), and the notion of a payoff of an external trait, which will be used when generalising condition (2). 

\subsection{$r$-Jacobian Matrix}

For each state $\a$, let $J_\a^r$ denote the Jacobian matrix at $\a$ with respect to the $r$-payoff function $z^r$ (henceforth, the \emph{$r-$Jacobian matrix}); that is, $J_\a^r$ is a square matrix of size $|supp(\a)|$, where the $ij$-th element in the matrix is the partial derivative of $z^r_{\a}(a_i)$ with respect to $\a(a_j)$:
\begin{equation}
(J_\a^r)_{a_i,a_j}=\frac{\partial z^r_\a(a_i)}{\partial \a(a_j)} \textrm{~~~~~~}\text{for all } a_i,a_j\in\supp(\a).
\end{equation}

Observe that when $r=0$ (the replicator dynamic) the $r$-Jacobian coincides with the payoff matrix (restricted to $\supp(\a)$) up to multiplication by the constant $u_\a$, i.e.: $$(J_\a^0)_{a_i,a_j}=\frac{\partial z^0_\a(a_i)}{\partial \a(a_j)}=\frac{\partial (u_\a(a_i)/u_\a)}{\partial \a(a_j)}=\frac{u(a_i,a_j)}{u_\a}\textrm{~~~~}\text{for all } a_i,a_j\in\supp(\a).$$
The main results of this section show that one can replace the payoff matrix in Condition (1) of Facts \ref{claim-necc-asymp-replic}--\ref{claim-suffic-asymp-replic} by the $r$-Jacobian when characterising asymptotic stability under the recombinator dynamics. Specifically, they show that a stationary state is internally stable if the $r$-Jacobian matrix is negative definite, and it is unstable if the matrix is not negative-semi-definite. This is demonstrated in the following example.
\addtocounter{example}{-1}
\begin{example}[continued] 
We arbitrarily order the strategies in $\supp(x^*)$ as $(\sc,\ac,\sd,\ad)$.  A simple numeric calculation shows that the $r$-Jacobian in state $x^*$ is:\\
\[
-\left(\begin{array}{cccc}
1.05 & 0.49 & 0.49 & 1.71\\
0.56 & 2.71 & 1.12 & 0.25\\
0.56 & 1.12 & 2.71 & 0.25\\
1.61 & 0.36 & 0.36 & 1.02
\end{array}\right)
\]

\noindent Observe that the  $r$-Jacobian matrix is not negative semidefinite with respect to the tangent space, which implies that $x^*$ is unstable. Specifically, let $w=(-1,0,0,1)\in T_{A}$ be a vector describing a small perturbation that slightly increases the share of $\ad$-agents and slightly decreases the share $\sc$-agents. Observe that $w^\intercal\cdot U \cdot w=(-0.56,0.13,0.13,0.69)\cdot w=1.25>0$, which implies that this small perturbation will take the population away from $x^*$. 
\hfill \myenddefinition
\end{example}

\subsection {Stable Partner Distribution and Invading Trait Payoff}
Consider an invasion of a population $\a$ by mutants bearing trait $a_d \notin \supp(\a)$.
In such a scenario, there is a qualitative difference between the replicator dynamics ($r=0$) and the recombinator dynamics with a positive recombination rate ($r>0$). 
Under the replicator dynamics, the mutant type carrying trait $a_d$, say type $a=(a_d,a_{-d})$ (where we refer to the trait profile $a_{-d}$ as the \emph{partner} of trait $a_d$), 
remains constant (that is, $a_d$ does not combine with other partners, hence the only type in the population bearing $a_d$ is $(a_d,a_{-d})$), and thus the success of trait $a_d$ in invading the population depends solely on the payoff of the mutant type combining the trait $a_d$ and its partner $a_{-d}$.

In contrast, under recombination the distribution of partners of an invading trait $a_d$ typically changes after a mutant carrying $a_d$ is introduced into a stationary population $\a$. In what follows we show that the distribution of partners converges towards a unique stable distribution of partners, which is independent of the specific initial mutant type that introduces $a_d$ to the population. We denote this stable distribution of partners of an invading trait $a_d$ by $\eta_\a^{a_d}$. The distribution  $\eta_\a^{a_d}$ induces each type $(a_d,a_{-d})$ (where $a_{-d}\in \supp(\a)$) with a frequency that is an $r$-weighted average of two elements: (1) its own payoff times its own frequency, and (2) the product of the frequencies of traits in $a_{-d}$. Formally,

\begin{definition}\label{def:partner-definition} Fix $r>0$, a stationary state $\a$, and a trait $a_d\notin \supp_d(\a)$.
Then $\eta_\a^{a_d}\in\Delta(\supp_{-d}(\a))$, which we call the \emph{stable partner  distribution} of trait $a_d$, is the unique solution to the following set of $|\supp_{-d}(\a))|$ equations: 
\begin{equation}
\eta(a_{-d})=(1-r)\frac{\eta(a_{-d})u_{x}(a_{d},a_{-d})}{\sum_{a_{-d}\in supp_{-d}(\a)}\eta(a_{-d}')u_{x}(a_{d},a_{-d}')}+r\prod_{a_{d'}\in a_{-d}}x(a_{d'}),
\label{eq:eta-definition}
\end{equation}
for each $a_{-d} \in 
\supp_{-d}(x)$.
\end{definition}

We begin by showing that the stable partner distribution is well-defined.
\begin{proposition}\label{prop:UniquePartner}
    Equation (\ref{eq:eta-definition}) admits a unique solution in $\Delta(\supp_{-d}(\a))$  for any $r>0$,  stationary state $\a$, and trait $a_d\notin \supp_d(\a).$
\end{proposition}
\begin{proof}
The result is immediate when $r=1$ (in which case $\eta(a_{-d})=\prod_{a_{d'}\in a_{-d}}x(a_{d'}))$.
We henceforth assume that $r<1$. A solution for Equation (\ref{eq:eta-definition}) clearly exists by Brouwer's Fixed Point Theorem. To prove uniqueness, we will show that, for a given $r$, $\a$, and $a_d$, there is a unique $z_0>0$, such that every solution $\eta$ of Equation (\ref{eq:eta-definition}) satisfies $$\sum\limits_{a_{-d}} \eta(a_{-d})u_\a(a_d,a_{-d})=z_0.$$ 
Assuming that this holds true, substituting $z_0$ back into Equation (\ref{eq:eta-definition}) we obtain the system of linear equations
\begin{gather*}
    \eta(a_{-d})=(1-r)\frac{\eta(a_{-d})u_\a(a_d,a_{-d})}{z_0}+r\prod\limits_{a_{d'}\in a_{-d}} \a(a_{d'}),
\end{gather*}
for each $a_{-d} \in \supp_{-d}(x)$,
which clearly has a unique solution. 
As this must hold for every solution $\eta$, uniqueness follows. 

To prove the existence of such a $z_0$, for each $a_{-d}$ and each $z>(1-r)u_\a(a_d,a_{-d})$, define $\eta_z(a_{-d})$ by:
\begin{equation}\label{Eq:eta_z}
\eta_z(a_{-d})=\frac{z\cdot r\cdot\prod_{a_{d'}\in a_{-d}}x(a_{d'})}{z-(1-r)u_\a(a_d,a_{-d})}.
\end{equation}

Let $\overline{z}=\max_{a_{-d}\in\supp_{-d}(\a)}\left((1-r)u_\a(a_d,a_{-d})\right)$.
Observe that  
$\eta_z(a_{-d})$ is well-defined for each $a_{-d}\in\supp_{-d}(\a)$ and each $z>\overline{z}$.
Further observe that $\eta_z(a_{-d})$ is decreasing in $z$ for each $a_{-d}$, hence the sum $h(z)=\sum\limits_{a_{-d}\in supp_{-d}(\a)}\eta_z(a_{-d})$ is decreasing. 
Letting $z\to \overline{z}$ from above takes $h(z)\to \infty$ in the limit, and in contrast as $z\to \infty$ one obtains $$h(z)\to r\sum\limits_{a_{-d}\in supp_{-d}(\a)}\prod_{a_{d'}\in a_{-d}}x(a_{d'})=r<1.$$ 
It follows that there is a unique $z_0>\overline{z}>0$ such that $h(z_0)=1$. 
Given a solution $\eta$ of Equation (\ref{eq:eta-definition}), if we set $z=\sum\limits_{a_{-d}\in supp_{-d}(\a)}\eta(a_{-d})u_\a(a_d,a_{-d})$ then solving the equation for $\eta$ we obtain $\eta=\eta_z$.
Furthermore, $z=z_0$ by the uniqueness of $z_0$. 
\end{proof}

The fact that the partner distribution of an invading trait
converges to a stable distribution (as shown in the  the proof of Theorem \ref{thm:asymptotic-sufficient}) allows us to define the payoff of an invading trait as the weighted average of the types carrying this mutant trait, where the weights are distributed according to the stable distribution of partners. Formally,
\begin{definition}
Fix a stationary state $\a$ and recombination rate $r$. 
The \emph{payoff of an invading trait} $a_d \notin \supp_d(\a)$ is defined as:
\begin{equation}
\label{eq:u_a_d-def}
u^r_{x}\left(a_{d}\right) = \sum_{a_{-d} \in supp_{-d}(x)} \eta_{x}^{a_{d}}\left(a_{-d}\right) \cdot u_{x}\left(a_{d},a_{-d}\right).
\end{equation}
\end{definition}

The stable partner distribution is degenerate if the stationary state is homogeneous (pure). 
For example, the stable partner distribution of trait $\dd$ with respect to the stationary state $\sc$ in the partially enforceable prisoner's dilemma assigns mass one to $\ss$. The following example demonstrates the stable partner distribution for a heterogeneous stationary state in a new game, which we call the emotional hawk dove game.
\begin{example}[Stable partner distribution in the emotional hawk dove game] 
\label{exam-emotional-HD}
Consider a bargaining interaction in which each player simultaneously makes two choices:
\begin{enumerate}
    \item being a ``hawk'' (trait $\hh$) or a ``dove'' (trait $\dd$) in the bargaining.
    \item being ``emotional'' (trait $\ee$), ``rational'' (trait $\rr$) or ``versatile'' (trait $\vv$) during the bargaining process (where the latter trait  allows transitions  from emotional to rational phases within the bargaining process).
\end{enumerate}
The two dimensional set of types includes $6=2\cdot3$ elements $A=\{\dr,\dv,\de,\hr,\hv,\he\}$. The basic game is hawk-dove: two players have to divide a surplus worth 100 between them. Two doves divide it equally (50 each). A hawk obtain a large share of 70 against a dovish opponent. Finally, when two hawks are matched bargaining often fails, and hence each player obtains a low payoff of 10. 

The payoffs of the basic hawk-dove game are modified by the choice of each player's emotional approach. Being emotional helps a hawkish player and adds two to her share of the surplus, while it harms a dovish player and reduces four units from her share. Similarly, being rational helps a dovish play (adds two units to her share) but harms a hawkish player (reduces four units from her share). Finally, being versatile does not affect a player's payoff.

Observe that the heterogeneous state $\a$ that assigns mass $50\%$ to type $\hv$ and the remaining mass of $50\%$ to type $\dv$ is stationary 
(and it is straightforward to show that the corresponding $r$-Jacobian matrix is negative definite for all $r$-s, which implies internal stability). Consider an invasion of this population by a rare mutant type with the emotional trait $\ee$. Substituting the example's parameters in Equation (\ref{eq:eta-definition}) yields the following value of $\eta_\a^\ee(h)$ as a function of $r$:
\begin{equation}
    \eta(\hh)=(1-r)\frac{42\eta(\hh)}{42\eta(\hh)+36(1-\eta(h))}+\frac{r}{2}\,\,\Rightarrow\,\,\eta_\a^\ee(\hh)(r)=\frac{\sqrt{169r^2-4r+4}-13r+2}{4}.
\end{equation}
Figure \ref{fig:stable-share} illustrates the stable share of hawkish partners of the invading trait $\ee$ as a function of the recombination rate $r$. When the recombination rate is close to zero, almost all the partners of $\ee$ are hawkish. This share is decreases in in $r$ and it converges to $50\%$ as the recombination rate converges to 1. The payoff of the invading trait is
\begin{equation}
u_{x}(\ee)=\eta_{x}^{\ee}(r)(\hh)\cdot u_{x}(\he)+\eta_{x}^{\ee}(r)(\dd)\cdot u_{x}(\de)=\eta_{x}^{\ee}(r)(\hh) \frac{72+12}{2}+\eta_{x}^{\ee}(r)(\dd) \frac{46+26}{2},
\end{equation}
which is larger than $u_\a=40$ iff $\eta_{x}^{\ee}(r)(\hh)>\frac{2}{3}\iff r<\frac{1}{6}.$  Thus the stationary state $\a$ is stable against an invasion of the trait $\ee$ if $r>\frac{1}{6}$ and unstable if $r<\frac{1}{6}.$
\begin{table}[h]
\begin{centering}
\caption{Payoff Matrix for Emotional Hawk-Dove Game}
\begin{tabular}{c|c|c|c|c|c|c|}
\multicolumn{1}{c}{} & \multicolumn{1}{c}{$\mathfrak{dr}$} & \multicolumn{1}{c}{$\mathfrak{dv}$} & \multicolumn{1}{c}{$\mathfrak{de}$} & \multicolumn{1}{c}{$\mathfrak{hr}$} & \multicolumn{1}{c}{$\mathfrak{hv}$} & \multicolumn{1}{c}{$\mathfrak{he}$}\tabularnewline
\cline{2-7} \cline{3-7} \cline{4-7} \cline{5-7} \cline{6-7} \cline{7-7} 
$\mathfrak{dr}$ & \textcolor{blue}{50},~\textcolor{red}{50} & \textcolor{blue}{52},~\textcolor{red}{48} & \textcolor{blue}{56},~\textcolor{red}{44} & \textcolor{blue}{36},~\textcolor{red}{64} & \textcolor{blue}{32},~\textcolor{red}{68} & \textcolor{blue}{30},~\textcolor{red}{70}\tabularnewline
\cline{2-7} \cline{3-7} \cline{4-7} \cline{5-7} \cline{6-7} \cline{7-7} 
$\mathfrak{dv}$ & \textcolor{blue}{48},~\textcolor{red}{52} & \textcolor{blue}{50},~\textcolor{red}{50} & \textcolor{blue}{54},~\textcolor{red}{46} & \textcolor{blue}{34},~\textcolor{red}{66} & \textcolor{blue}{30},~\textcolor{red}{70} & \textcolor{blue}{28},~\textcolor{red}{72}\tabularnewline
\cline{2-7} \cline{3-7} \cline{4-7} \cline{5-7} \cline{6-7} \cline{7-7} 
$\mathfrak{de}$ & \textcolor{blue}{44},~\textcolor{red}{56} & \textcolor{blue}{46},~\textcolor{red}{54} & \textcolor{blue}{50},~\textcolor{red}{50} & \textcolor{blue}{30},~\textcolor{red}{70} & \textcolor{blue}{26},~\textcolor{red}{74} & \textcolor{blue}{24},~\textcolor{red}{76}\tabularnewline
\cline{2-7} \cline{3-7} \cline{4-7} \cline{5-7} \cline{6-7} \cline{7-7} 
$\mathfrak{hr}$ & \textcolor{blue}{64},~\textcolor{red}{36} & \textcolor{blue}{66},~\textcolor{red}{34} & \textcolor{blue}{70},~\textcolor{red}{30} & \textcolor{blue}{10},~\textcolor{red}{10} & \textcolor{blue}{6},~\textcolor{red}{14} & \textcolor{blue}{4},~\textcolor{red}{16}\tabularnewline
\cline{2-7} \cline{3-7} \cline{4-7} \cline{5-7} \cline{6-7} \cline{7-7} 
$\mathfrak{hv}$ & \textcolor{blue}{68},~\textcolor{red}{32} & \textcolor{blue}{70},~\textcolor{red}{30} & \textcolor{blue}{74},~\textcolor{red}{26} & \textcolor{blue}{14},~\textcolor{red}{6} & \textcolor{blue}{10},~\textcolor{red}{10} & \textcolor{blue}{8},~\textcolor{red}{12}\tabularnewline
\cline{2-7} \cline{3-7} \cline{4-7} \cline{5-7} \cline{6-7} \cline{7-7} 
$\mathfrak{he}$ & \textcolor{blue}{70},~\textcolor{red}{30} & \textcolor{blue}{72},~\textcolor{red}{28} & \textcolor{blue}{76},~\textcolor{red}{24} & \textcolor{blue}{16},~\textcolor{red}{4} & \textcolor{blue}{12},~\textcolor{red}{8} & \textcolor{blue}{10},~\textcolor{red}{10}\tabularnewline
\cline{2-7} \cline{3-7} \cline{4-7} \cline{5-7} \cline{6-7} \cline{7-7} 
\end{tabular}
\par\end{centering}
\end{table}

\begin{figure}[h]
\caption{Stable Share of Hawkish Partners $\eta_\a^\ee(\hh)(r)$ in Example \ref{exam-emotional-HD}}
\label{fig:stable-share}
\begin{centering}
\includegraphics[scale=0.35]{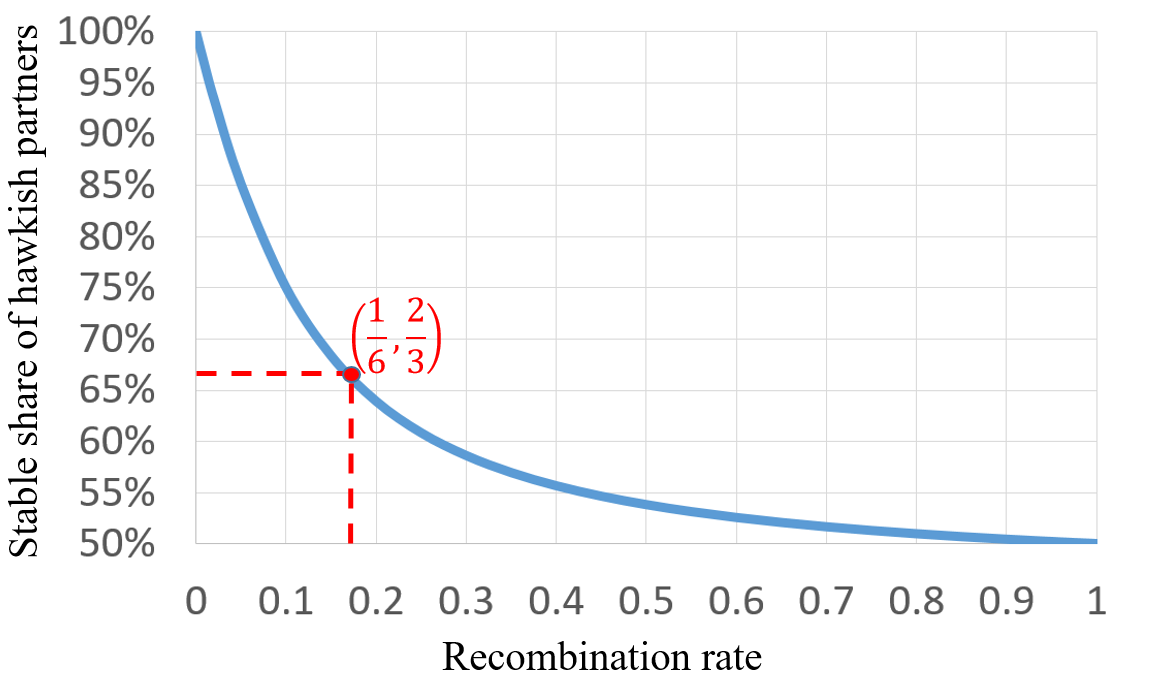}
\par\end{centering}
\end{figure}
\hfill \myenddefinition
\end{example}

\subsection{Main Results}
In what follows we characterise the set of asymptotically stable states. The characterisation extends Facts \ref{claim-necc-asymp-replic}--\ref{claim-suffic-asymp-replic} by showing that the asymptotic stability of  a stationary state is characterised by the following three conditions:
\begin{enumerate}
   \item \emph{Internal stability:} The $r$-Jacobian matrix is (semi-)negative-definite
   with respect to the tangent space.
      \item \emph{External stability against traits:} The payoffs of traits outside the support are lower than the payoffs of the traits of the incumbents.
    \item \emph{External stability against types:} The payoffs of types outside the support are at most $1-r$ times higher than the average payoffs of the incumbents. 
    \end{enumerate}
    
    Our main results are, essentially, an if and only if characterisation of asymptotic stability. Specifically, the strict variants of these conditions  (i.e., negative-definiteness and strictly lower payoffs of external traits/types) imply asymptotic stability, and their weak counterparts (semi-negative-definiteness and weakly lower payoffs of external traits/types) are implied by
    Lyapunov stability (which is implied by asymptotic stability). Formally,

\begin{theorem}
\label{thm:asymptotic-neccesary}
If a state $x^*$ is 
Lyapunov stable with recombination rate $r$, then it satisfies
\begin{enumerate}
     \item Weak internal stability:  $w^\intercal\cdot J_{\a^*}^r w\leq 0$ for each $w\in T_{\supp({\a^*})}$.
    \item External stability against traits: $u_{\a^*} \geq u^r_{\a^*}(a_d)$ for each $a_d \notin \supp_d({\a^*})$.
    \item External stability against types: $u_{\a^*} \geq (1-r)u_{\a^*}(a)$ for each $a \notin \supp({\a^*})$.
\end{enumerate}
\end{theorem}

\begin{theorem}
\label{thm:asymptotic-sufficient}
A stationary state ${\a^*}$ is asymptotically stable with recombination rate $r$ if it satisfies
\begin{enumerate}
     \item Internal stability: $w^\intercal\cdot J_{\a^*}^r w< 0$ for each $w\in T_{\supp({\a^*})}$.
    \item External stability against traits: $u_{\a^*} > u^r_{\a^*}(a_d)$ for each $a_d \notin \supp_d({\a^*})$.
    \item External stability against types: $u_{\a^*} > (1-r)u_{\a^*}(a)$ for each $a \notin \supp({\a^*})$.
\end{enumerate}
\end{theorem}

\begin{proof}[Sketch of proof for Theorem \ref{thm:asymptotic-neccesary}] (See Appendix \ref{app:Proof1and2} for the full formal proof.)
\begin{enumerate}
    \item Suppose to the contrary that the $r$-Jacobian matrix $J_{x^*}^r$ is not negative semi-definite. 
    This implies that there exists a state $y$ satisfying $\supp{y}\subset\supp{\a^*}$, such that $(\a^*-y)'J_{\a^*}^r(\a^*-y)>0.$ This, in turn, implies that starting from a perturbed  state $\epsilon\cdot y+(1-\epsilon)\cdot \a^*$ giving weight $\epsilon$ to $y$ for sufficiently small $\epsilon>0$, the average payoff according to distribution $y$ is strictly greater than the weighted average payoff according to $\a^*$, which implies by $r$-monotonicity that the dynamic of the population will move away from $\a^*$.
    
    \item Suppose to the contrary that $u_{\a^*} < u^r_{\a^*}(\hat a_d)$ for some $\hat a_d \notin \supp_d(\a^*)$. Define $\hat\eta_{\a^*}^{\hat a_d}\in\Delta(A)$:
    \[
\hat{\eta}_{\a^*}^{\hat{a}_{d}}\left(a\right)=\begin{cases}
\eta_{\a^*}^{\hat{a}_{d}}\left(a_{-d}\right) & a_{d}=\hat{a}_{d}\textrm{ and }a_{-d}\in\supp_{-d}(\a^*)\\
0 & \textrm{otherwise}
\end{cases}
\]
    Consider a perturbed state $(1-\e)\a^*+\e\hat{\eta}_{\a^*}^{\hat{a}_{d}}$ for sufficiently small $\e>0.$ 
    Proposition \ref{prop:UniquePartner} implies that $\hat{a}_{d}$'s partners are distributed according to the stable partner distribution $\eta_{\a^*}^{\hat{a}_{d}}$, and that this distribution remains stationary as long as $\epsilon \ll 1$. 
    This implies that the mean payoff of agents with trait $\hat{a}_{d}$ is $u^r_{\a^*}(\hat a_d)>u_{\a^*}.$  Trait-payoff monotonicity (Proposition \ref{prop-trait-stationary}) then implies that the share of trait  $\hat{a}_{d}$ keeps increasing, and hence that the population eventually moves away from $x$.
    
    \item  Suppose to the contrary that $u_{\a^*} < (1-r)u_{\a^*}(\hat a)$ for $\hat a \notin \supp(\a^*).$ Consider the slightly perturbed state $(1-\e)\a^*+\e \hat a$. The fact that $u_{\a^*} < (1-r)u_{\a^*}(\hat a)$ implies that $z^r_{\a^*}(\hat a)>1$ (see Equation (\ref{eq:zed_def})), which implies, due to Proposition \ref{prop-1}, that the share of agents playing action $\hat a$ keeps growing, and hence that the population moves away from $x$.\qedhere
\end{enumerate}
\end{proof}

\begin{proof}[Sketch of proof for Theorem \ref{thm:asymptotic-sufficient}] (See Appendix \ref{app:Proof1and2} for the formal proof.)

Assume to the contrary that $\a^*$ is not asymptotically stable. Then for every $\epsilon>0$ there is an $\e$-nearby state $y_\e$ (i.e., satisfying $|y_\e-\a^*|<\e$) such that a population starting at $y_\e$ does not converge to $\a^*$. 
The fact that the $r-$Jacobian matrix is negative definite implies that $y_\e(a)\rightarrow \a^*(a)$ for each $a \in \supp(\a^*)$, for sufficiently small $\e$.
The inequality $u_{\a^*} > (1-r)u_{\a^*}(a)$ for each $a \notin \supp(\a^*)$ implies that the share of new agents who imitate any  $a \notin \supp(\a^*)$ as a whole block (i.e., without recombination) converges to zero (see Equation (\ref{eq:zed_def})). 

Finally, the inequality  $u_{\a^*} > u^r_{\a^*}(a_d)$ for each $a_d \notin \supp_d(\a^*)$, implies, due to trait monotonicity (Proposition \ref{prop-2}) that the share of agents who play each external trait $a_d$ converges to zero. 
Combining these facts, we see that the population converges from any sufficiently nearby state $y_\e$ to $\a^*$, which implies that $\a^*$ is asymptotically stable.
\end{proof}

Next we apply Theorems \ref{thm:asymptotic-neccesary}--\ref{thm:asymptotic-sufficient} to characterise stable states in two interesting special cases: (I) pure states, and (II) interior states. 

Corollaries \ref{cor-pure-nec}--\ref{cor-pure-suf} (which extend \citeauthor{waldman1994systematic}'s (\citeyear{waldman1994systematic}) Proposition 2) show that a pure state $a$ is asymptotically stable essentially iff the payoff of $a$ against itself is higher than (1)  the payoff of any "neighbouring" strategy (which differs in a single trait) against $a$, and (2) $1-r$ times the payoff of any strategy (which might differ in multiple traits). Formally
\begin{corollary}\label{cor-pure-nec}
If a pure state $a\in A$ is Lyapunov stable with recombination rate $r$, then:
\begin{enumerate}
    \item $u_a(a)\geq u_a(a'_d,a_{-d})$ for any dimension $d\in D$ and any trait $a'_d\in A_d,$ and
    \item $u_a(a)\geq (1-r)\cdot u_a(a')$ for any type $a'\in A.$
\end{enumerate}
\end{corollary}

\begin{corollary}\label{cor-pure-suf}
Pure stationary state $a\in A$ is asymptotically stable with recombination rate $r$ if:
\begin{enumerate}
    \item $u_a(a)> u_a(a'_d,a_{-d})$ for any dimension $d\in D$ and any trait $a_d\neq a'_d\in A_d,$ and
    \item $u_a(a)> (1-r)\cdot u_a(a')$ for any type $a\neq a'\in A.$
\end{enumerate}
\end{corollary}
\addtocounter{example}{-3}
\begin{example}[revisited]
Corollary \ref{cor-pure-suf} implies that any strict Nash equilibrium is asymptotically stable for any recombination rate $r$. This implies that $\ad$ is asymptotically stable for all values of $r$ in the partially enforceable prisoner's dilemma. Observe that $\sc$ yields a higher payoff against itself than both of its neighbouring strategies (namely, $\sd$ and $\ac$). This implies that $\sc$ is asymptotically stable if $$15=u_\sc(\sc)> (1-r)\cdot u_\sc(\ad)=16 \iff r>\frac{1}{16},$$ and it is unstable if $r<\frac{1}{16}.$
\end{example}
Corollary \ref{cor-interior} shows that an interior (full support) state $x$ is asymptotically stable essentially iff the $r$-Jacobian matrix is negative definite. Formally
\begin{corollary}\label{cor-interior}
\begin{enumerate}
    \item If an interior state $x^*$ is Lyapunov stable with recombination rate $r$, then $x^TJ_{x^*}^r x\leq0$ for each $x\in\Delta(A)$.
    \item An interior stationary state $a\in A$ is asymptotically stable with recombination rate $r$ if $x^TJ_{x^*}^r x<0$ for each $x\in\Delta(A)$.
\end{enumerate}
\end{corollary}
Corollaries \ref{cor-pure-nec}--\ref{cor-interior} follow immediately from Theorems \ref{thm:asymptotic-neccesary}--\ref{thm:asymptotic-sufficient}.

\section{A General View}\label{sec-extension}

The results of the previous section were obtained in the context of the particular equation of motion expressed by the recombinator dynamic of Equation \ref{eq:recombinatorMotion}.
We consider in this section a broader class of equations of motion, 
and analyse stationarity and stability in this more general setup.
\subsection{Extended Model}
Say that a function $f: A \times \Delta(A) \to [0,1]$ is a 
\emph{regular} function if, for each fixed $x \in \Delta(A)$, the equality 
$\sum_{a \in A} f(a,x) = 1$ holds. This is equivalent to saying that for each $x$, the vector $f(\cdot,x) - x(\cdot)$ is a tangent vector of the simplex $\Delta(A)$,
which implies that the $x^t$ always remains in the state space $\Delta(A)$.%
\footnote{\ 
Our definition of regularity is analogous to Definition 4.1 in \cite{weibull1997evolutionary}. In our canonical example, the replicator dynamics, we have $f(a,x) = \frac{\a(a)u_\a(a)}{u_\a}$. } 
In what follows, we will sometimes write $f(a|x)$ as a synonym for $f(a,x)$, when $f:A \times \Delta(A) \to [0,1]$.

Given functions $f_1, f_2: A \times \Delta(A) \to [0,1]$, and $r \in [0,1]$ consider dynamics of the form
\begin{equation}
	\label{eq:generalRecombinatorMotion}
    \dot{x}(a) = (1-r)f_1(a|x) + rf_2(a|x) - x(a).
\end{equation}

For example, consider Equation (\ref{eq:recombinatorMotion}).
Define $f_1(a|x) = \frac{u_x(a)}{u_x} x(a)$ and $f_2(a|x)=\prod_{a_d \in a} \frac{\a(a_d)u_\a(a_d)}{u_\a}$.
It is easy to show that 
$f_1$ and $f_2$ are regular functions and hence Equation (\ref{eq:recombinatorMotion}) is a special case of the form of Equation (\ref{eq:generalRecombinatorMotion}).
We shall, from now on, 
assume that $f_1(a|x)$ and $f_2(a|x)$
refer to regular functions s
(which implies that $x^t$ remains in the simplex for any $t$).

\subsection{Characterising Stationary States} 

Define the $r$-payoff of action $a \in \supp(\a)$ at state $x$ as:
\begin{equation}
	\label{eq:gen-zeta_def}
    \zeta^{r}_\a(a) \:= \frac{1} {x(a)} \biggl((1-r)f_1(a|x) + rf_2(a|x)\biggr).
\end{equation}
Dividing Equation (\ref{eq:generalRecombinatorMotion}) by $x(a)$ for $a \in \supp(\a)$ now yields
\begin{equation}
	\label{eq:generalRelativeGrowth1}
	\frac{\dot{\a}(a)}{\a(a)} = \zeta^{r}_\a(a) - 1,  
\end{equation}
which generalises Equation (\ref{eq:recombDynamicsSimple}).
\begin{proposition}
	\label {prop:stationary-general-recombinator}
	The population recombinator dynamic of Equation (\ref{eq:generalRecombinatorMotion}) satisfies:
	\begin{enumerate}
		\item $r$-payoff monotonicity:  $\zeta^{r}_\a(a) > \zeta^{r}_\a(a') \Leftrightarrow \frac{\dot\a(a)}{\a(a)}>\frac{\dot\a(a')}{\a(a')}$ for each state $\a \in \Delta(A)$ and for each pair of types $a,a' \in \supp(\a)$.
		\item A state $\a$ is stationary if and only if $\zeta^{r}_\a(a) = 1$ for all $a \in \supp(\a).$
	\end{enumerate}
\end{proposition}

\begin{proof}
	The proof is identical to the proof of Proposition \ref{prop-1}, with the necessary changes such as using Equation (\ref{eq:generalRelativeGrowth1}) in place of Equation (\ref{eq:recombDynamicsSimple}), etc.
\end{proof}

In even greater generality, instead of restricting to two regular functions alone, it is possible to work with any integer $m \ge 2$ regular functions and recapitulate the entire argument of this section.
We leave this straightforward generalisation to the reader.
\begin {example} Consider a learning dynamic in which agents occasionally revise their actions at a constant rate, normalised to one. Each revising agent randomly chooses a mentor with a probability that is proportional to that mentor's fitness. With probability $1-r$ the new agent imitates all of the traits of the mentor, while with the remaining probability of $r$, the new agent randomly chooses only a single dimension and imitates the mentor's trait in this dimension (while keeping the agent's existing traits fixed in all other dimensions). 

This learning process induces dynamics that fit Equation (\ref{eq:generalRecombinatorMotion}), in which 
$f_1(a|x) = \frac{u_x(a)}{u_x} x(a)$ (as in the baseline model) and $f_2(a|x)= \frac{1}{n} \sum_{d\in D} \frac{x(a_{-d}) x(a_{d}) u_{x}(a_{d})}{u_{x}}$. Substituting this in Equation (\ref{eq:gen-zeta_def}) yields the following $r$-payoff:
\begin{equation}
    \zeta^{r}_\a(a) \:= (1-r)\frac{u_x(a)}{u_x} + \frac{r}{n\cdot x(a)} \sum_{d\in D} \frac{x(a_{-d}) x(a_{d}) u_{x}(a_{d})}{u_{x}}.
\end{equation}
This example illustrates the flexibility of the general approach in applicability to a range of possible learning dynamics. 
\hfill \myenddefinition
\end{example}

\subsection{Trait-Centric Dynamics}
Another perspective is attained by noting that the constitution of $A$ as $A =  A_1 \times \ldots \times A_n$ enables the definition, for any trait $a_d \in A_d$, of the marginal frequency $x(a_d)$ and the marginal payoff $u_x(a_d)$.
An equation of motion of the form of Equation (\ref{eq:generalRecombinatorMotion}) determines trajectories in $\Delta(A)$. 
Denoting $\Delta_i \:= \Delta(A_i)$, via projection operators one may consider which trajectories in $\Delta_1 \times \ldots \times \Delta_n$ are induced by the trajectories in $\Delta(A)$ determined by $\Phi_u$.
In what follows we generalise the notion of trait payoff monotonicity. 
We need a way of measuring the `strength' of a trait $a_d$ in dimension $D$ relative to any other trait in the same dimension, given a regular function $f$ and state $x$.
The key to achieving this is by looking at the partners in $A_{-d}$, because $a_d$ alone gets no payoff: it needs partners.
For each partner $a'_{-d} \in A_{-d}$, a trait $a_d$ receives payoff $f(a_d,a'_{-d},x)$.
The measure of the relative strength of $a_d$ is then given by the sum of its payoffs for all possible partners.


\begin{definition}
Let $f(a,x)$ be a regular function.
Define the \textit{marginal function} of $f$ to be
\begin{equation}
	\varphi(a_d,x)=
\frac{1}{x(a_d)}\sum_{a'_{-d} \in A_{-d}} f(a_d,a'_{-d},x).
\end{equation}

We say that $\varphi$ is \emph{trait payoff increasing} if $u_x(a_d)>u_x(\widehat{a}_d)$ implies $\varphi(a_d,x)>\varphi(\widehat{a}_d,x)$ for every state $x$ and every pair of traits $a_d,\widehat{a}_d\in A_d$. 
\hfill \myenddefinition
\end{definition}

Denote by $\varphi_1$ the marginal function of $f_1(\cdot|x)$ and by $\varphi_2$ the marginal function of $f_2(\cdot|x)$.  
Then dividing Equation (\ref{eq:generalRecombinatorMotion}) by $x(a_d)$ and summing over all $a'_{-d} \in A_{-d}$ yields the \emph{trait-centric recombinator dynamics}:
\begin{align}
\label{eq:generalRecombinatorForm}
	\frac{\dot{\a}(a_d)}{\a(a_d)} &= \sum_{a'_{-d} \in A_{-d}} \frac {\dot{\a}(a_d, a'_{-d})}{x(a_d)}  \\ 
 &= \frac{1-r}{x(a_d)}\sum_{a'_{-d} \in A_{-d}} f_1(a_d,a_{-d}'|x) 
 + \frac{r}{x(a_d)}\sum_{a'_{-d} \in A_{-d}} f_2(a_d,a_{-d}'|x) - 1 \nonumber \\
	&= (1-r) \varphi_1(a_d,x)+ r \varphi_2(a_d,x) - 1. 
 \nonumber
\end{align}

\begin{example} 
Consider again Equation (\ref{eq:recombinatorMotion}), which was previously shown to be an equation of motion involving regular functions. In that case it follows from Equation (\ref{eq:recombDynamics-trait}) that
\begin{gather}
    \varphi(a_d,x)=u_\a(a_d) \biggl(\frac{(1-r)} {u_x}   
	+ \frac{r} {u_x^{|D|}}  \sum_{a'_{-d} \in A_{-d}} \prod_{a'_{d'} \in a'_{-d}} \a(a'_{d'})u_\a(a'_{d'}) \biggr).
\end{gather}
Notice that if $u_\a(a_d)>u_\a(\widehat{a}_d)$ then indeed $\varphi(a_d,x) > \varphi(\widehat{a}_d,x)$, 
hence the function $\varphi$ defined for the recombinator is indeed trait payoff increasing.
It follows that Equation (\ref{eq:recombDynamics-trait}) is a special case of the form of Equation (\ref{eq:generalRecombinatorForm}).
\hfill \myenddefinition
\end{example}

\begin{proposition}
    
	\label{prop:marginalPayoffMonotonicity}

	(1) Let $\varphi_1$ and $\varphi_2$ be trait payoff 
	increasing. Then
    \begin{enumerate}
    \item 
	The trait-centric recombinator dynamic of Equation (\ref{eq:generalRecombinatorForm}) satisfies trait payoff monotonicity, that is, $u_\a (a_d) > u_\a (a'_d) \Leftrightarrow \frac{\dot\a(a_d)}{\a(a_d)}>\frac{\dot\a(a'_d)}{\a(a'_d)}$ for each state $\a \in \Delta(A)$ and for each pair of traits $a_d,a'_d \in \supp_d(\a)$.
	
	\item  
    A state $x$ is a stationary state only if $u_\a (a_d) = u_\a (a'_d)$ for each pair of traits $a_d,a'_d \in \supp_d(\a)$.
    \end{enumerate} 
\end{proposition}

\begin{proof}
    See Appendix \ref{app:proof-of-Prop-5}.
\end{proof}

\subsection{Characterising Stable States} 
 Consider an equation of motion of the form 
 \begin{equation}
 \label{eq:traitImitation}
     \dot{x}(a)=(1-r)f_1(a|x)+rf_2(a|x)-x(a).
 \end{equation}
 We assume that $f_1$ and $f_2$ are non-negative, regular and trait-payoff-increasing functions. 
 We further assume that $f_1$ implements a \emph{type-imitation dynamic} in the sense that $f_1(a|x)=g_1(a|x)\cdot x(a)$, where $g_1(a|x)$ is differentiable and satisfies the property of \textit{monotone percentage growth rate}, namely $g_1(a|x)>g_1(a'|x)$ if and only if $u_x(a)>u_x(a')$ (this property is typically called `imitation dynamics' in the literature, see, e.g., \cite[Example 1]{sandholm2010Local}). 
 Finally, we assume the following three assumptions on  $f_2$, 
  which together with our previous assumptions imply that it is a \emph{trait-imitation dynamic}:
  \begin{enumerate}
      \item[1.] \emph{Differentiability}: $f_2$ is differentiable;
      \item[2.] 
      \emph{Trait combination}: $\prod\limits_{d\in D}x(a_d)>0$ implies $f_2(a|x)>0$;
      \item[3.] \emph{Trait growth inertia}: $x(a_d)=0$ and $\frac{\partial f_2(a_d,a_{-d}|x)}{\partial x(a')}> 0$ implies 
      $\prod\limits_{d'\neq d} x(a_{d'})> 0$ and $a'_d=a_d.$  
  \end{enumerate}

The assumption of trait combination implies that the function $f_2$ combines traits into types; that is, if all traits that are part of a type $a$ exist in the population, then there will be a positive share of new agents who will adopt type $a$.

To understand the idea behind \textit{trait growth inertia}, 
suppose that $x$ is a stationary state, and suppose that trait $a_d$ is absent from the population. In such a case Equation (\ref{eq:traitImitation}) implies (by linear approximation) that near a stationary state the growth in a type is determined by the change of $f_2(a|x)$ in the directions of the various types $a'$. 
Which types matter here? 
Trait growth inertia implies that if new agents adopt a type $a$, this can happen only if $a_d$ is the only trait in type $a$ that does not exist in the population, and that the perturbation in the population has been a result of invasion of trait $a_d$ via some type $a'$ carrying it. \vskip 5pt 

If all of the above assumptions hold, then we say an equation of motion of the form of Equation (\ref{eq:traitImitation}) expresses \emph{generalised recombinator dynamics}, which generalises Equation (\ref{eq:recombinatorMotion}), where the first component  captures type imitation (generalising $\frac{\a(a)u_\a(a)}{u_\a}$) and the second component  captures trait imitation (generalising $\prod\limits_{d\in D} \frac{x(a_d)u_x(a_d)}{u_x}$).
Denote by $J_x$ the restriction to $\supp(x)\times\supp(x)$ of the Jacobian of the function 
\begin{gather*}
    x\mapsto\left((1-r)f_1(\cdot|x)+rf_2(\cdot|x)-x(\cdot)\right).
\end{gather*} 

 For $a_{-d}'\in\supp_{-d}(x^*)$, define $v_{x^*}(a_d|a_{-d}')$ as the trait imitation rate for $a_d$ with respect to partners $a_{-d}'$ (generalising the relative payoff $\frac{u_{x^*}(a_d,a_{-d}')}{u_{x^*}}$ within the second summand in Equation (\ref{eq:Linear-Approximation-Traits}) of Appendix \ref{app:LinearApproximation}), in detail:
  $$v_{x^*}(a_d|a_{-d}')=\frac{\frac{\partial f_2(a_d,a_{-d}'|x^*)}{\partial x(a_d,a_{-d}')}}{\prod\limits_{d'\neq d}x^*(a_{d'}')}.$$ 
 To illustrate, in the recombinator dynamics $f_2(a|x)=\prod\limits_{d\in D} \frac{x(a_d)u_x(a_d)}{u_x}$, thus since $\frac{u_{x^*}(a_{d'}')}{u_{x^*}}=1$ for $d'\neq d$ we have $\frac{\partial f_2(a_d,a_{-d}'|x^*)}{\partial x(a_d,a_{-d}')}=\left(\prod\limits_{d'\neq d} x^*(a_{d'}')\right)\frac{u_{x^*}(a_d,a_{-d}')}{u_{x^*}}$, hence $v_{x^*}(a_d|a_{-d}')=\frac{u_{x^*}(a_d)}{u_{x^*}}$.
 
 Given a distribution $y_{a_d}\in\Delta(\supp_{-d}(x^*))$ we further define $u_{x^*}(a_d|y_d)$ as the marginal payoff of trait $a_d$ resulting from type-imitation (generalising the marginal trait-to-partner payoff from Appendix \ref{app:PartnerDistribution}) 
 $$u_{x^*}(a_d|y_{a_d})=\sum\limits_{a_{-d}'\in \supp_{-d}(x^*)} g_1(a_d,a_{-d}'|x^*)y_{a_d}(a_{-d}').$$ 
 We define  $v_{x^*}(a_d|y_{a_d})$ as the marginal payoff of trait $a_d$ resulting from trait imitation (generalising the average appearing in the second summand in Equation (\ref{eq:Linear-Approximation-Traits}) in Appendix \ref{app:LinearApproximation}):
$$v_{x^*}(a_d|y_{a_d})=\sum\limits_{a_{-d}'\in \supp_{-d}(x^*)}v_{x^*}(a_d|a_{-d}')y_{a_d}(a_{-d}').$$  

In our analysis we shall consider the marginal payoff of trait $a_d$ resulting from \emph{generalised recombinator dynamics} given by $$U_{x^*}^r(a_d|y_{a_d})=(1-r)u_{x^*}(a_d|y_{a_d})+rv_{x^*}(a_d|y_{a_d}).$$ The latter will play a  role in the generalised partner dynamics, which is a dynamic on $\Delta(\supp_{-d}(x^*))$ defined by 
\begin{align}
    \dot{y}_{a_d}(a_{-d}) = (&1-r)g_1(a_d,a_{-d}|x^*)y_{a_d}(a_{-d})\\
    &+r\biggl(\prod\limits_{d'\neq d}x^*(a_{d'})\biggr)v_{x^*}(a_d|y_{a_d})- y_{a_d}(a_{-d})U_{x^*}^r(a_d|y_{a_d}),
    \nonumber
\end{align}

We say that the \emph{generalised partner dynamics} has a unique globally stable equilibrium if for every trait $a_d$ in dimension $d$ there is a state $\eta_{x^*}^{a_d}\in\Delta(\supp_{-d}(x^*))$ such that $y_{a_d}\to \eta_{x^*}^{a_d}$ as $t \to \infty$, independently of the initial conditions. 

With these assumptions we can now 
extend our characterisation of stable states:
\begin{theorem}
\label{thm:general-asymptotic-neccesary}
Let $x^*$ be a stationary state of generalised recombinator dynamics such that the partner-dynamics has a unique 
globally stable equilibrium $\eta_{x^*}^{a_d}$ for every dimension $d$ and trait $a_d$. If the $x^*$ is  Lyapunov stable, then it satisfies
\begin{enumerate}
     \item Weak internal stability:  $w^\intercal\cdot J_{x^*} w\leq 0$ for each $w\in T_{\supp({x^*})}$.
    \item External stability against traits: $1 \geq U^r_{x^*}(a_d|\eta_{x^*}^{a_d})$ for each $a_d \notin \supp_d({\a^*})$.
    \item External stability against types: $1 \geq (1-r)f_1(a|x^*)$ for each $a \notin \supp({x^*})$.
\end{enumerate}
\end{theorem}

\begin{theorem}
\label{thm:-general-asymptotic-sufficient}
Let $x^*$ be a stationary state of  generalised recombinator dynamics such that the partner-dynamics has a unique globally stable equilibrium $\eta_{x^*}^{a_d}$ for every dimension $d$ and trait $a_d$. The stationary state ${\a^*}$ is asymptotically stable if it satisfies
\begin{enumerate}
    \item Internal stability: $w^\intercal\cdot J_{x^*} w< 0$ for each $w\in T_{\supp({x^*})}$.
    \item External stability against traits: $1 > U^r_{x^*}(a_d|\eta_{x^*}^{a_d})$ for each $a_d \notin \supp_d({\a^*})$.
    \item External stability against types: $1 > (1-r)f_1(a|x^*)$ for each $a \notin \supp({\a^*})$.
\end{enumerate}
\end{theorem}

In Appendix \ref{app:UniqueGloballyStable} we describe a family of functions, containing the recombinator dynamics, for which the induced partner dynamics admits a unique globally stable equilibrium. 

\section{Conclusion}\label{sec-conclusion}
In this paper we began with a simple observation: nearly all replicator models of evolutionary game theory in the current literature suppose that every agent in a population can be associated with a single fitness-determinative trait. This, however, is far from realistic; we are all composed of complex ensembles of traits and behaviours, and frequently our success is critically dependent on many or all of the elements of those ensembles. 

We therefore expanded the standard models to one in which there is a set of dimensions, where each dimension is itself a set of traits, and an agent is associated with a tuple of traits, one from each dimension. In what we called a combinator model, newly born agents sample tuples of incumbent agents, each of whom is a mentor, and imitate a trait from each mentor.
In a recombinator model, a proportion $0 \le r \le 1$ of the population samples multiple mentors by the combinator and $1-r$ selects only one mentor, as in the replicator.

The result is a rich and interesting dynamic social learning model, predicting novel behaviour in various applications. 
The recombinator dynamics violate payoff monotonicity when the recombination rate is positive, and this can allow strictly dominated types to be asymptotically stable.
To analyse the dynamics, we defined the $r$-payoff function, which combines the effects of the combinator and replicator components of the dynamics into a single recombinator vector field.
In trajectories flowing along this vector field, payoff monotonicity is restored, indicating that this is the natural payoff function on which to focus in this setting.

There are in fact actually two parallel games intertwined in the model: at a visible level, agents interacting in a game determining population-dependent payoffs, alongside a less immediately obvious but significant game being played by competing traits.  The importance of the trait-centric perspective emerges here, as at the stationary state of a convergent trajectory surviving types may exhibit different payoffs, but at the traits level, the dynamics reliably select for traits with higher payoffs.

Stability is similarly studied, with the recombinator model generalising known results on asymptotic stability in the replicator model. However, while asymptotic stability in the replicator model depends on internal and external stability to invading types, the recombinator adds an extra layer of complexity: external stability now needs to obtain against both invading traits and invading types to ensure stability.  

There is much scope for future research to develop these ideas in several directions, including developing a multi-population version of the recombinator dynamics, establishing microfoundations by revision protocols, and studying parallel models based on other evolutionary dynamics, such as best reply, logit, and Brown–von Neumann–Nash.

\begin{adjustwidth}{0pt}{-21pt}
\begin{singlespace}
\bibliographystyle{chicago}
\bibliography{mybibdata}
\end{singlespace}
\end{adjustwidth}

\newpage
\appendix
\pagenumbering{arabic}
\section*{Online Appendix -- For Online Publication}
\section{Formal Proofs of Theorems \ref{thm:asymptotic-neccesary} and \ref{thm:asymptotic-sufficient}}
\label{app:Proof1and2}

\subsection{Proof Outline}
The proof of Theorems \ref{thm:asymptotic-neccesary} and \ref{thm:asymptotic-sufficient} proceeds as follows. We consider three different cases. 
The first case involves the dynamics of types in the support of an equilibrium. 
The second case involves types that have traits outside the support in multiple dimensions.
The third and last case involves types that have traits outside the support in a single dimension.

For each case we begin by considering the linear approximation of the dynamics near an equilibrium. In the first case we obtain a condition akin to that of the replicator dynamics, 
thus internal stability allows us to prove local stability. 
In the last case we show that the approximate dynamics is given by a diagonal matrix, and external stability for types implies negative values along the diagonal. 

The most challenging case is the second one. 
We first show that the linearisation of the dynamics there can be presented as linear dynamics over components whose partners have positive proportion in equilibrium plus some ``external force'' emerging from the interaction with types with unsupported partners. It turns out that if the linear component over supported partners is negative definite then the dynamics as a whole is stable. To prove that it is indeed negative definite we introduce the notion of ``partner dynamics'' and prove that these dynamics lead to a unique stable equilibrium which is in fact the stable partner distribution introduced in Equation (\ref{eq:eta-definition}). We then use this fact to show the negative definiteness of the original linearisation.

\subsection{The Linear Approximation}\label{app:LinearApproximation}

We consider in this section a linear approximation of the recombinator dynamics around a stationary state $\a^*$. The Lyapunov--Poincar\'e Theorem \cite[page 41]{Bellman} implies that the stability properties of the recombinator dynamics at a stationary state $x^*$ will be implied by the stability properties of its linear approximation around the state $x^*$. 

1. \emph{Internal Stability}. 
Note that if one restricts attention to $\supp(\a^*)$ then the recombinator equation is a 
replicator-like equation with respect to the $r$-payoff field $z^r$
(in the sense that the form
$\frac{\dot{\a}(a)}{\a(a)} = z^r_\a (a) - 1$ relates $\frac{\dot{\a}(a)}{\a(a)}$ on one side of the equation to a function of $a$ on the other side).
Hence expanding around $\a^*$ we obtain, since 
$\sum\limits_{a\in supp(\a)} \a(a)z_\a^r(a)=1$, that 
(neglecting terms that are $O(||x-x^*||^2))$ for any $a\in \supp(x^*)$:
\begin{equation}
\label{eq:LinearApproxInner}
    \dot{\a}(a)=\a^*(a)\sum\limits_{a'\in \supp(\a)} \frac{\partial z_{\a^*}^r(a)}{\partial \a(a')}(\a(a')-\a^*(a'))=\a^*(a)\sum\limits_{a'\in \supp(\a)} (J_{\a^*}^r)_{a,a'}(\a(a')-\a^*(a')).
\end{equation}

From Equation (\ref{eq:LinearApproxInner}) it is immediately clear that it is sufficient for $J_{\a^*}^r$ to be negative definite (with respect to the tangent space) for local asymptotic stability of the dynamics of $\a(a)$ for $a\in \supp(\a^*)$ to obtain. 
In the other direction, suppose that $\a^*$ is Lyapunov stable.
Then $J_{\a^*}^r$ must be negative-semi-definite: any positive value in $J_{\a^*}^r$ would initiate a trajectory that would eventually carry the state of the population away from $\a^*$, contradicting Lyapunov stability. 

2. \emph{External Stability Against Types}.
Next consider a type $a \not \in \supp(\a^*)$. Notice that as the combinator part of the recombinator equation is a product of averages over traits $a_d\in a$, the linear factors are obtained by types $a'$ such that $\a^*(a'_d)=0$ for a single value of $d$, namely -- only for types in which only one of their traits is outside the support of $x^*$.

If $a\in A$ has multiple traits outside the support $x^*$, we obtain the approximate linear dynamic (neglecting terms that are $O(||x-x^*||^2))$
\begin{equation}
\label{eq:LinearApproxTypes}
    \dot{\a}(a)=\a(a)\biggl((1-r)\frac{u_{\a^*}(a)}{u_{\a^*}}-1\biggr).
\end{equation}
It is clear from Equation (\ref{eq:LinearApproxTypes}) that it suffices for $(1-r)\frac{u_{\a^*}(a)}{u_{\a^*}}-1<0$ for the weight of $a$ to be reduced asymptotically (at an exponential rate), hence the same condition suffices for the local asymptotic stability of $\a$ restricted to neighbour types to hold.
In the other direction, if $(1-r)\frac{u_{\a^*}(a)}{u_{\a^*}}-1 >0$, then Lyapunov stability cannot obtain, as the weight of $a \not \in \supp(\a^*)$ increases, carrying the population away from $\a^*$. 

3. \emph{External Stability Against Traits}. 
We first consider the proof of Theorem \ref{thm:asymptotic-neccesary} for this part. Assume the dynamics is Lyapunov stable and that external-stability-against-traits doesn't hold, namely that $u_{\a^*}^r(a_d)-u_{\a^*} > 0$ for some $a_d\in \supp_d(x^*)$. 
 By appealing again to the monotonicity property proved in Proposition \ref{prop-trait-stationary}, it follows that the weight $\a^t(a_d)$ increases monotonically, hence moving the population away from $\a^*$, contradicting the assumption of Lyapunov stability.\vskip 5pt 

The proof of local stability (Theorem \ref{thm:asymptotic-sufficient}) for this part is more involved and will be carried in several stages. As before, we begin with the linear approximation for traits $a_d$ in dimension $d$.
Suppose that $a_{-d}\in\supp_{-d}(x^*)$, while $(a_d,a_{-d})\not\in\supp(x^*)$. 
The linear approximation for the equation of motion for $a_{-d}\in\supp_{-d}(x^*)$ is (neglecting terms that are $O(||x-x^*||^2))$
\begin{align}
\label{eq:Linear-Approximation-Traits}
    \dot{x}(a_d,a_{-d})
     = (&1-r)\frac{u_{x^*}(a_d,a_{-d})}{u_{x^*}}x(a_d,a_{-d})\\ \nonumber &+\frac{r}{u_{x^*}}\prod\limits_{d'\neq d}x^*(a_{d'})\sum\limits_{a_{-d}'\in A_{-d}}u_{x^*}(a_d,a_{-d}')x(a_d,a_{-d}')-x(a_d,a_{-d}).
\end{align}

Note that if $a_{-d}'\not\in\supp_{-d}(x^*)$ then (as discussed after Equation (\ref{eq:LinearApproxTypes}) above)
\begin{gather}\label{eq:Non-neighbours-die-fast}
    x^t(a_d,a_{-d}')=x^0(a_d,a_{-d}')e^{\left((1-r)\frac{u_{x^*}(a_d,a_{-d}')}{u_{x^*}}-1\right)t},
\end{gather} 
where $x^0$ represents the initial state and $x^t$ the state at time $t > 0$.

Hence, setting 
\[
f_{a_d}(t)=\frac{r}{u_{x^*}}\prod\limits_{d'\neq d} x^*(a_{d'})\sum\limits_{a_{-d}'\not\in\supp_{-d}(x^*)}u_{x^*}(a_d,a_{-d}')x^0(a_d,a_{-d}')e^{\left((1-r)\frac{u_{x^*}(a_d,a_{-d}')}{u_{x^*}}-1\right)t}
\]
we can rewrite Equation (\ref{eq:Linear-Approximation-Traits}) for $a_{-d}\in\supp_{-d}(x^*)$ as 
\begin{align}
\label{eq:Linear-Approximation-Traits-Reduced1}
    \dot{x}(a_d,a_{-d}) = (&1-r)\frac{u_{x^*}(a_d,a_{-d})}{u_{x^*}}x(a_d,a_{-d})
    \\ \nonumber &+\frac{r}{u_{x^*}}\prod\limits_{d'\neq d}x^*(a_{d'})\sum\limits_{a_{-d}'\in \supp_{-d}(x^*)}u_{x^*}(a_d,a_{-d}')x(a_d,a_{-d}')-x(a_d,a_{-d})+f_{a_d}(t),
\end{align}
noting that $\int_0^
\infty f_{a_d}(s)ds<\infty$. 

Consider the matrix $\Lambda^{a_d}$ indexed by elements $a_{-d},a_{-d}'\in \supp_{-d}(x^*)$ given by (which presents Equation (\ref{eq:Linear-Approximation-Traits-Reduced1}) for a fixed $a_d$ in matrix form):
\begin{gather}\label{eq:The-Matrix}
    \Lambda^{a_d}_{a_{-d},a_{-d}'}=\begin{cases} \left((1-r)+r\prod\limits_{d'\neq d}x^*(a_{d'})\right)\frac{u_{x^*}(a_d,a_{-d})}{u_{x^*}}-1, & a_{-d}=a_{-d}'\\
    \frac{ru_{x^*}(a_d,a_{-d}')}{u_{x^*}}\prod\limits_{d'\neq d}x^*(a_{d'}), & a_{-d}\neq a_{-d}'
    \end{cases}. 
\end{gather}
Using $\Lambda^{a_d}$, and letting $f(t)$ be the vector whose entries equal $f_{a_d}(t)$, Equation (\ref{eq:Linear-Approximation-Traits-Reduced1}) can be rewritten more compactly as
\begin{gather}\label{eq:Linear-Approximation-Traits-Reduced2}
    \dot{x}(a_d,\cdot)=\Lambda^{a_d} x(a_d,\cdot) +f(t),
\end{gather}
where $\int_0^\infty ||f(s)||ds<\infty$.
The following lemma is then a direct corollary of 
\citet[Theorem 2]{Brauer1964}:
\begin{lemma}\label{lem:Brauer}
 If $\Lambda^{a_d}$ is negative definite then all the solutions of Equation (\ref{eq:Linear-Approximation-Traits-Reduced2}) tend to $0$ as $t\to\infty$, and $0$ is in fact Lyapunov stable. 
\end{lemma} 
 
 \begin{proof}
    We prove Lemma \ref{lem:Brauer} by applying \cite[Theorem 2]{Brauer1964}. First we notice that Equation (\ref{eq:Linear-Approximation-Traits-Reduced2}) has the form of  \cite[Equation (1)]{Brauer1964}. Indeed, in terms of \cite[Theorem 2]{Brauer1964} and \cite[Equations (5) and (9)]{Brauer1964}, set $A=\Lambda^{a_d}$,  $f(t,y)=f_{a_d}(t)$, and $p(t)\equiv 0$. 
    
    Therefore, it is sufficient to prove that if $A=\Lambda^{a_d}$ is negative definite then $||e^{At}||\leq Ke^{-\sigma t}$ for some $\sigma>0$. This is an immediate corollary of the implication $(3)\Rightarrow (2)$ in Chicone (1999) \cite[Theorem 2.34]{Chicone1999}. 
\end{proof} \qedhere

Lemma \ref{lem:Brauer} shows that to complete the proof we should show that $\Lambda^{a_d}$ is negative definite. We will do that in the following section by introducing the notion of \emph{partner distribution} and showing how it may be used together with the reduction above to prove stability.

 \subsection{The Partner Distribution and the Proof of the Main Results}\label{app:PartnerDistribution}

We now begin to examine the negative definiteness of $\Lambda^{a_d}$. To prove that it will be sufficient to prove that $x(a_d,\cdot)=0$ is locally stable for the differential equation $\dot{x}(a_d,\cdot)=\Lambda^{a_d}x(a_d,\cdot)$, which by definition is the same as Equation (\ref{eq:Linear-Approximation-Traits-Reduced1}) with the terms $f_{a_d}(t)$ neglected. 

First, sum Equation (\ref{eq:Linear-Approximation-Traits-Reduced1}), neglecting the terms $f_{a_d}(t)$, over $a_{-d}\in\supp_{-d}(x^*)$:
\begin{gather}\label{eq:SupportLinearApproximateTraits}
     \dot{x}(a_d)=(1-r)\frac{\sum\limits_{a_{-d}\in \supp_{-d}(x^*)}x(a_d,a_{-d})u_{x^*}(a_d,a_{-d})}{u_{x^*}}\\ \nonumber+r\frac{\sum\limits_{a_{-d}'\in \supp_{-d}(x^*)} u_{x^*}(a_d,a_{-d}')x(a_d,a_{-d}')}{u_{x^*}}-x(a_d).
 \end{gather}
  
 Consider the \textit{partner distribution} $y(a_{-d})\:=\frac{x(a_d,a_{-d})}{x(a_{d})}$ for $a_{-d}\in\supp_{-d}(x^*)$, which is the ratio of the distribution of $a_d$ to the weight of $a_d$ at state $x$. Define the \emph{normalised marginal trait-to-partner payoff} $u_{x^*}(a_d|y)=\sum\limits_{a_{-d}\in \supp_{-d}(x^*)}y(a_{-d})\frac{u_{x^*}(a_d,a_{-d})}{u_{x^*}}$. From Equation (\ref{eq:SupportLinearApproximateTraits}) we obtain the differential equation: 
\begin{gather}\nonumber
     \dot{x}(a_d)=(1-r)u_{x^*}(a_d|y)x(a_d)+ru_{x^*}(a_d|y)x(a_d)-x(a_d)\\
     \label{eq:SupportLinearApproximateTraitsPartnersNon} =(u_{x^*}(a_d|y)-1)x(a_d)
 \end{gather}

 \begin{lemma}\label{lem:Flow-Convergence}
  Suppose that for every initial condition $x^0$ in an open neighbourhood of $x^*$ the partner distribution $y^t\to \eta_{x^*}^{a_d}$ as $t\to \infty$ and furthermore $u_{x^*}(a_d|\eta_{x^*}^{a_d})-1<0$ as implied by internal stability. Then $x(a_d)=0$ is a locally stable equilibrium of Equation (\ref{eq:SupportLinearApproximateTraitsPartnersNon}) 
 \end{lemma}

 \begin{proof}
   There is a $C>0$ and time $t_0>0$ s.t. for every $t>t_0$ we have $u_{x^*}(a_d|y^t)-1<-C<0$. Integrating Equation (\ref{eq:SupportLinearApproximateTraitsPartnersNon}) we have, for some constant $M$ and every $t>t'>t_0$
 \begin{gather}
     x^t(a_d)= Me^{\int_{t'}^t (u_{x^*}(a_d|y^s)-1)ds}<Me^{-C(t-t')}.
 \end{gather}
 Taking $t\to \infty$ implies the Lemma. 
 \end{proof}

 We have thus shown that if the flow of the partner distribution $y$ (as given by solution of the differential equation $\dot{x}(a_d,\cdot)=\Lambda^{a_d}x(a_d,\cdot)$) converges to the stable partner distribution for every initial state near $x^*$ and furthermore internal stability holds then stability follows in this case as well. 
 
 In the last two sections of this Appendix, we will show that the partner distribution follows a certain dynamic and that this dynamic has the stable partner distribution as a unique globally stable equilibrium, which will finish the proof of Theorem \ref{thm:asymptotic-sufficient}.

 \subsection{Partner Dynamics}\label{app:PartnerDynamic}

In this section we prove that the partner distribution, given by the solution of the differential equation $\dot{x}(a_d,\cdot)=\Lambda^{a_d}x(a_d,\cdot)$, follows a dynamic. In the next and final section we will prove that the stable partner distribution is its unique globally stable equilibrium.  

 Since $y(a_{-d})x(a_d)=x(a_d,a_{-d})$ one has $\dot{y}(a_{-d})x(a_d)+y(a_{-d})\dot{x}(a_d)=\dot{x}(a_d,a_{-d})$. Dividing by $x(a_d)$ and rearranging we obtain 
 \begin{gather}
     \dot{y}(a_{-d})=\frac{\dot{x}(a_d,a_{-d})}{x(a_d)}-y(a_{-d})\frac{\dot{x}(a_d)}{x(a_d)}.
 \end{gather}
 Substituting this into Equation (\ref{eq:Linear-Approximation-Traits-Reduced1}) (neglecting the $f_{a_d}(t)$ term), and considering the normalised payoff $\widehat{u}_{x^*}(a)=\frac{u_{x^*}(a)}{u_{x^*}}$ we obtain 
 \begin{gather}\label{eq:PartnerDynamics}
     \dot{y}(a_{-d})=(1-r)\widehat{u}_{x^*}(a_d,a_{-d})y(a_{-d})+r\left(\prod\limits_{d'\neq d} x^*(a_{d'})\right)u_{x^*}(a_{d}|y)-y(a_{-d})u_{x^*}(a_{d}|y),
 \end{gather}
 where, by definition $u_{x^*}(a_d|y)=\sum\limits_{a_{-d}'\in\supp_{-d}(x^*)} \widehat{u}_{x^*}(a_d,a_{-d}')y(a_{-d}').$

 Going back to Lemma \ref{lem:Flow-Convergence} in Appendix \ref{app:PartnerDistribution}, we are only left to prove that the partner dynamic is globally stable and that $y^t\to \eta_{x^*}^{a_d}$ as $t\to \infty$, where $\eta_{x^*}^{a_d}$ is the stable partner distribution. We prove this in the next section, Appendix \ref{app:StabilityPartner}. \vskip 5pt

\subsection{Global Stability of the Partner Dynamics}
\label{app:StabilityPartner}

Recall that for a trait $a_d$ and an associated partner distribution $y\in\Delta(\supp_{-d}(x^*))$ we defined the partner dynamics as
\begin{equation}
\label{eq:ApproxPartnerDynamics}
    \dot{y}(a_{-d})=(1-r)\widehat{u}_{\a^*}(a_d,a_{-d})y(a_{-d})+ru_{\a^*}(a_d|y)\prod\limits_{d'\neq d} \a^*(a_{d'})-y(a_{-d})u_{\a^*}(a_d|y),
\end{equation}
where $\widehat{u}_{x^*}(a)=\frac{u_{x^*}(a)}{u_{x^*}}$. Abusing notation (or alternatively, normalising the payoff units by the factor $u_{x^*}$) we will henceforth denote $u_{x^*}(a)$ instead of $\widehat{u}_{x^*}(a)$. 

Note that $\eta$ is a stationary point of this dynamics if and only if 
\begin{gather}\label{eq:StablePartnersRecall}
    (1-r)\frac{u_{\a^*}(a_d,a_{-d})\eta(a_{-d})}{u_{\a^*}(a_d|\eta)}+r\prod\limits_{d'\neq d} \a^*(a_{d'})-\eta(a_{-d})=0,
\end{gather}
which means that stationary states of the partner dynamic of a trait correspond to the stable partner distribution as defined in Definition \ref{def:partner-definition}.
 
We now turn to prove the proposition on global stability. First notice that the sign of right-hand-side of Equation (\ref{eq:ApproxPartnerDynamics}) is the same as that of the left-hand-side of Equation (\ref{eq:StablePartnersRecall}). 
    Fixing $a_{-d}$, let $\eta^0$ be a distribution satisfying $\eta^0(a_{-d})=0$, and for $t\in[0,1]$ set $\eta^t(a_{-d})=t$ and $\eta^t(a_{-d}')=(1-t)\eta^0(a_{-d}')$ for $a_{-d}'\neq a_{-d}$. 
    Plugging $\eta^t$ instead of $\eta$ into Equation (\ref{eq:StablePartnersRecall}) we obtain the equation 
    \begin{gather}
        (1-r)\frac{t u_{\a^*}(a_d,a_{-d})}{tu_{\a^*}(a_d,a_{-d})+(1-t)u_{\a^*}(a_d|\eta^0)}+r\prod\limits_{d'\neq d} \a^*(a_{d'})-t=0.
    \end{gather}
    Notice that the left hand side of this last equation, which we denote $g_{a_{-d}}(t|\eta^0)$ or $g(t)$ for short, is either a concave or a convex function of $t$. Furthermore, $g(0)=r\prod\limits_{d'\neq d}x^*(a_{d'})>0$, and $g(1)=1-r+r\prod\limits_{d'\neq d}x^*(a_{d'})-1<0$, which implies that $g$ changes sign exactly once for every choice of distribution $\eta^0$.  
  
  Denoting the left-hand-side of Equation (\ref{eq:StablePartnersRecall}) by $H_{a_{-d}}(\eta)$, we can now deduce that the sets $V_{a_{-d}}^+=\{\eta: H_{a_{-d}}(\eta)>0\}$ and $V_{a_{-d}}^-=\{\eta: H_{a_{-d}}(\eta)<0\}$ are two open, disjoint, and connected sets. If $\eta(a_{-d})=0$ we have
  $H_{a_{-d}}(\eta)>0$, and if $\eta(a_{-d})=1$ we have $H_{a_{-d}}(\eta)<0$. We deduce that $|\eta^t(a_{-d})-\eta_{x^*}^{a_d}(a_{-d})|$ decreases monotonically to $0$ regardless of the initial point $\eta^0$, hence the dynamics satisfy the property of Lyapunov stability. 
  
  Furthermore, let $\eta^*$ be a limit point of the dynamic, namely, suppose that there is a sequence $t_n\to \infty$ such that $\y^{t_n}\to \eta^*$. 
  Then it must hold that $H_{a_{-d}}(\eta^*)=0$ for every $a_{-d}\in A_{-d}$, hence by the uniqueness of $\eta_\a^{a_d}$ we have $\eta^*=\eta_{\a}^{a_d}$. 
 \vskip 5pt

 \section{General View: Proofs and a Detailed Example}\label{app:GeneralView}

 \subsection{Proof of Proposition \ref{prop:marginalPayoffMonotonicity}}\label{app:proof-of-Prop-5}
    Suppose that $u_\a (a_d) > u_\a (a'_d)$ for each pair of traits $a_d,a'_d \in \supp_d(\a)$. 
    Then by the assumption of trait payoff monotonicity, $\varphi_1(a_d,x)>\varphi_1({a'}_d,x)$ and $\varphi_2({a}_d,x) > \varphi_2({a'}_d,x)$ for every state $x$ and every pair of traits $a_d,{a'}_d \in \supp_d(\a)$.
    Hence by Equation (\ref{eq:generalRecombinatorForm}) it follows that $\frac{\dot\a(a_d)}{\a(a_d)}>\frac{\dot\a(a'_d)}{\a(a'_d)}$.\vskip 3pt
    
    In the other direction, suppose that $\frac{\dot\a(a_d)}{\a(a_d)}>\frac{\dot\a(a'_d)}{\a(a'_d)}$ for all $a_d,a'_d \in \supp_d(\a)$ and every state $\a$.
    Then by Equation (\ref{eq:generalRecombinatorForm}) $(1-r) \varphi_1(a_d,x)+r\varphi_2(a_d,x) > (1-r) \varphi_1({a'}_d,x)+r\varphi_2({a'}_d,x)$ for all $a_d,a'_d \in \supp_d(\a)$ and every state $\a$.
    
    Suppose now that $u_\a (a_d) \le u_\a (a'_d)$ for some pair $a_d,a'_d \in \supp_d(\a)$.
    Then by the assumption that $f_1$ and $f_2$ satisfy trait payoff monotonicity it follows that $\varphi_1(a_d,x) \le \varphi_1(a'_d,x)$ and $\varphi_2(a_d,x) \le \varphi_2(a'_d,x)$, which further implies $(1-r) \varphi_1(a_d,x)+r\varphi_2(a_d,x) \le (1-r) \varphi_1(a'_d,x)+r\varphi_2(a'_d,x)$. 
    This is a contradiction establishing the conclusion that $u_\a (a_d) > u_\a (a'_d)$.
    
    With monotonicity established, from here standard arguments from monotonicity establish that $x$ is a stationary state only if $u_\a (a_d) = u_\a (a'_d)$ for each pair of traits $a_d,a'_d \in \supp_d(\a)$.

\subsection{Proof of Theorems  \ref{thm:general-asymptotic-neccesary} and \ref{thm:-general-asymptotic-sufficient}}

The proof of Theorem \ref{thm:general-asymptotic-neccesary} follows the exact footsteps of Theorem \ref{thm:asymptotic-neccesary}, where the role of Proposition \ref{prop-trait-stationary} is replaced by Proposition \ref{prop:marginalPayoffMonotonicity}.

We turn to explain how to adapt the proof of Theorem \ref{thm:asymptotic-sufficient} to prove Theorem \ref{thm:-general-asymptotic-sufficient}. Suppose that $x^*$ is a stationary point and $f_2$ satisfies conditions (1) - (3) (of the statements of the theorems) at $x^*$. Then for $a\not\in \supp(x^*)$, as in the proof of Theorems 1 and 2, $x^*(a_d)=0$ for at least one dimension $d\in D$. If $x^*(a_d)=x^*(a_{d'})=0$ for $d\neq d'$ then condition (3) implies the linear approximation in this case is 
\begin{gather}
    \dot{x}(a)\approx \left[(1-r)f_1(a|x^*)-1\right]x(a),
\end{gather}
and external stability for traits will determine stability here. 

Suppose that $x^*(a_d)=0$ for only one dimension $d\in D$. Recalling that $f_1(a|x)=g_1(a|x)x(a)$, the linear approximation is 
\begin{gather}\label{eq:ExtendedRecombinator}
    \dot{x}(a)\approx(1-r)g_1(a|x^*)x(a)+r\sum\limits_{a_{-d}'\in A_{-d}}\frac{\partial f_2(a_d,a_{-d}'|x^*)}{\partial x(a_d,a_{-d}')}x(a_d,a_{-d}')- x(a_d,a_{-d}).
\end{gather}
Recall that we have defined $v_{x^*}(a_d,a_{-d}')=\frac{\frac{\partial f_2(a_d,a_{-d}'|x^*)}{\partial x(a_d,a_{-d}')}}{\prod\limits_{d'\neq d}x^*(a_{d'})}$. 
Then Equation (\ref{eq:ExtendedRecombinator}) becomes 
\begin{gather}\label{eq:ExtendedRecombinator2}
    \dot{x}(a)\approx(1-r)g_1(a|x^*)x(a)+r\prod\limits_{d'\neq d}x^*(a_{d'})\sum\limits_{a_{-d}'\in A_{-d}}v_{x^*}(a_d,a_{-d}')x(a_d,a_{-d}')- x(a_d,a_{-d}).
\end{gather}
As in the Equations (\ref{eq:Non-neighbours-die-fast}) and (\ref{eq:Linear-Approximation-Traits-Reduced1}) in the proof of Theorems 1 and 2, we can write 
\begin{align}
\label{eq:ExtendedRecombinator3}
    \dot{x}(a) = (&1-r)g_1(a|x^*)x(a) \\
    \nonumber
    &+r\prod\limits_{d'\neq d}x^*(a_{d'})\sum\limits_{a_{-d}'\in \supp_{-d}(x^*)}v_{x^*}(a_d,a_{-d}')x(a_d,a_{-d}')- x(a_d,a_{-d})+f_{a_d}(t),
\end{align}
where $f_{a_d}$ is a linear combination of strictly negative exponential functions. 

As in the proof of Theorems 1 and 2, Equation (\ref{eq:ExtendedRecombinator3}) has the form $\dot{x}(a_d,\cdot)=\Lambda^{a_d}x(a_d,\cdot)+f(t)$, where $f$ is a vector whose coordinates are the functions $f_{a_d}$ and $\Lambda^{a_d}$ is a matrix constructed in the same manner as in Equation (\ref{eq:The-Matrix}) in Appendix \ref{app:LinearApproximation}. Lemma \ref{lem:Brauer} from that appendix shows again that to complete the proof we should prove that $\Lambda^{a_d}$ is negative definite. We will now explain how the proof in Appendix \ref{app:Proof1and2} can be adapted to our generalisation. 

We sum over $a_{-d}\in \supp_{-d}(x^*)$ on both sides of the equation $\dot{x}(a_d,\cdot)=\Lambda^{a_d} x(a_d,\cdot)$, namely Equation (\ref{eq:ExtendedRecombinator3}) with the $f_{a_d}$ term neglected. Recall that we have defined $u_{x^*}(a_d|y_{a_d})=\sum\limits_{a_{-d}'\in \supp_{-d}(x^*)} g_1(a_d,a_{-d}'|x^*)y_{a_d}(a_{-d}')$, and $v_{x^*}(a_d|y_{a_d})=\sum\limits_{a_{-d}'\in \supp_{-d}(x^*)}v_{x^*}(a_d,a_{-d}')y_{a_d}(a_{-d}')$, where $y_{a_d}\in\Delta(\supp_{-d}(x^*))$ is the partner distribution of $a_d$. Summing over partners $a_{-d}$ we thus obtain 
\begin{gather}
    \dot{x}(a_d)=\left[(1-r)u_{x^*}(a_d|y_{a_d})+rv_{x^*}(a_d|y_{a_d}) -1\right]x(a_d).
\end{gather}

Recall that we have defined  $U_{x^*}^r(a_d|y_{a_d})=(1-r)u_{x^*}(a_d|y_{a_d})+rv_{x^*}(a_d|y_{a_d})$, hence we obtain the equation 
\begin{gather}\label{eq:TraitGeneral}
    \dot{x}(a_d)=\left(U_{x^*}^r(a_d|y_{a_d}) -1\right)x(a_d).
\end{gather}
As argued following Equation (\ref{eq:SupportLinearApproximateTraitsPartnersNon}) in Appendix \ref{app:PartnerDistribution}, to prove local stability here it sufficient to show that the partner distribution $y_{a_d}$, defined by the solutions of the differential equation $\dot{x}(a_d,\cdot)=\Lambda^{a_d} x(a_d,\cdot)$, follows a dynamic and that this dynamic has stable partner distribution $\eta_{x^*}^{a_d}$ as its unique globally stable equilibrium. Next we derive the form of this generalised partner dynamics, which will finish the proof as its global stability is assumed as a condition, as a straightforward generalisation of our proof in Appendices \ref{app:PartnerDynamic} and \ref{app:StabilityPartner}.

\subsubsection{The Form of the Generalised Partner Dynamics}
We now turn to deriving the partner distribution dynamics in this case. Recall that $\dot{y}_{a_d}(a_{-d})=\frac{\dot{x}(a_d,a_{-d})}{x(a_d)}-y_{a_d}(a_{-d})\frac{\dot{x}(a_d)}{x(a_d)}$.
This leads to 
\begin{align}
\label{eq:GeneralPartnerDynamics}
    \dot{y}_{a_d}(a_{-d}) = (&1-r)g_1(a_d,a_{-d}|x^*)y_{a_d}(a_{-d})\\
    &+r\biggl(\prod\limits_{d'\neq d}x^*(a_{d'})\biggr)v_{x^*}(a_d|y_{a_d})- y_{a_d}(a_{-d})U_{x^*}^r(a_d|y_{a_d}).
    \nonumber
\end{align}
A stationary point $\eta$ of this equation must therefore satisfy 
\begin{gather}
    0=(1-r)\frac{g_1(a_d,a_{-d}|x^*)}{U_{x^*}^r(a_d|\eta)}\eta(a_{-d})+r\frac{\biggl(\prod\limits_{d'\neq d}x^*(a_{d'})\biggr)v_{x^*}(a_d|\eta)}{U_{x^*}^r(a_d|\eta)}- \eta(a_{-d}).
\end{gather}

If, assuming external stability for traits, a unique globally stable stationary distribution $\eta_{x^*}^{a_d}$ exists here as assumed in the statement of our theorems, the stability argument follows in the same manner as in Lemma \ref{lem:Flow-Convergence} in Appendix \ref{app:PartnerDistribution}.

\subsection {Partner Dynamics with a Unique Globally Stable Equilibrium}\label{app:UniqueGloballyStable}
In this appendix we describe a family of functions containing the recombinator dynamics for which the induced partner dynamics admits a unique globally stable equilibrium.

In this family of dynamics the trait-imitation component will be a product of (1) the frequencies of all traits, and (2) a function $G$ that only depends on the relative trait utilities. We begin by formally presenting this functional form for $f_2$ and proving that it is indeed a trait-imitation dynamics.

    \begin{proposition}\label{prop:Trait-Imitation}
        An equation of motion of the form of Equation (\ref{eq:traitImitation}) in which\\ $f_2(a|x)=\left(\prod\limits_{d\in D} x(a_d)\right) G\left(\left(\frac{u_x(a_d)}{u_x}\right)_{d\in D}\right)$, where $G>0$ is continuously differentiable, satisfies the three conditions characterising  trait imitation dynamics. 
    \end{proposition}

\begin{proof}
We prove that each of the properties defining trait-imitation hold. 

For differentiability, suppose $x^*(a_d)=0$ while $x^*(a_{d'})>0$ for every $d'\neq d$.  We can write 
    \begin{gather}
      \frac{\partial f_2(a|x^*)}{\partial x(a')}=\lim\limits_{h\to 0}\frac{f_2(a|x^*+he_{a'})}{h}
      =\lim\limits_{h\to 0}\frac{1}{h}\int_0^h\frac{d}{dh}\left[f_2(a|x^*+he_{a'})\right]dh. 
    \end{gather}
    Notice that if $a_d\not\in a'$ then the limit above is $0$. 
    For the sake of notational convenience, we denote $\textbf{u}_x(a)=\left(\frac{u_x(a_{d'})}{u_x}\right)_{d'\in D}$. Thus, we consider $a'=(a_d,a_{-d}')$. Notice that \begin{gather*}
        f_2(ax^*+se_{(a_d,a_{-d}')})=s\left(\prod\limits_{d'\neq d}x(a_{d'})\right)G(\textbf{u}_{x^*+se_{(a_d,a_{-d}')}}(a)).
    \end{gather*}
    Furthermore, $\frac{u_{x^*+se(a_d,a_{-d}')}(a_d)}{u_{x^*+se(a_d,a_{-d}')}}=\frac{s^{-1}u_{x^*+se(a_d,a_{-d}')}(a_d,a_{-d}')s}{u_{x^*}+O(s)}=\frac{u_{x^*}(a_d,a_{-d}')}{u_{x^*}}+o(s)$. Thus we obtain, as $s\to 0$,
    \begin{gather*}
        f_{2}(a|x^*+se_{(a_d,a_{-d}')})=\left(\prod\limits_{d'\neq d}x(a_{d'})\right)G\left(\frac{u_{x^*}(a_d,a_{-d}')}{u_{x^*}}, \left(\textbf{u}_{x^*}(a)\right)_{-d}\right)\cdot s+o(s),
    \end{gather*}
    hence $f_2$ is differentiable. 
       For the other two properties, if $f_2(a|x)=0$ then $\prod\limits_{d\in D}x(a_d)=0$, proving nullity. If $x(a_d)=0$ then $\prod\limits_{d'\neq d}x(a_{d'})=0$ or $a_d'\neq a_d$ implies $\frac{\partial f_2(a_d,a_{-d}|x)}{\partial x(a')}=0$ proving trait dependence. \qedhere \vskip 5pt

\end{proof}
    
    This family of dynamics is akin to imitation dynamics with a single trait (and indeed they coincide if $|D|=1$; see \cite[Example 1]{sandholm2010Local}). 
    Observe that the recombinator dynamics has this functional form, where $G((u_d)_{d\in D})=\prod\limits_{d\in D}u_d$.
      
   We next assume that $G$ is symmetric (in the sense of being invariant to permutations), and we assume that the type-imitation component has the following form (for an arbitrary dimension $d$):
     \begin{gather*}
       f_1(a|x)=x(a)G\left(\frac{u_{x^*}(a)}{u_{x^*}}, \textbf{1}_{-d}\right),
   \end{gather*}    
       namely, that $g_1(a|x)=G\left(\frac{u_{x^*}(a)}{u_{x^*}}, \textbf{1}_{-d}\right)$. Observe that under these assumptions $v_{x^*}$ and $u_{x^*}$ coincide. Indeed, if $x^*(a_d)=0$ and $a_{-d}'\in\supp_{-d}(x^*)$ then $\frac{u_{x^*}(a_{d'}')}{u_{x^*}}=1$ by Proposition \ref{prop:marginalPayoffMonotonicity}, hence $\frac{\partial f_2(a_d,a_{-d}'|x^*)}{\partial x(a_d,a_{-d}')}=\left(\prod\limits_{d'\neq d} x^*(a_{d'}')\right) G\left(\frac{u_{x^*}(a_d,a_{-d}')}{u_x^*},\textbf{1}_{-d}\right)$, and therefore :
     \begin{gather*}
        v_{x^*}(a_d|a_{-d}')=G\left(\frac{u_{x^*}(a_d,a_{-d}')}{u_x^*},\textbf{1}_{-d}\right)=g_1(a|x), 
    \end{gather*}
     hence
         \begin{gather*}
       v_{x^*}(a_d|y_{a_d})=\sum\limits_{a_{-d}'\in \supp_{-d}(x^*)} g_1(a_d,a_{-d}'|x^*)y_{a_d}(a_{-d}')=u_{x^*}(a_d|y_{a_d}).
   \end{gather*}
     Recall that $v_{x^*}(a_d|y_{a_d})$ is the marginal payoff of trait $a_d$ resulting from trait-imitation, while $u_{x^*}(a_d|y_{a_d})$ is the marginal payoff of trait $a_d$ resulting from type-imitation. Thus, the equality above implies that trait-imitation and type-imitation equilibrate. In this case, $U_{x^*}^r(a_d|\cdot)=u_{x^*}(a_d|\cdot)=v_{x^*}(a_d|\cdot)$.
   We will prove that in such case 
   the partner dynamics has a unique globally stable stationary state.

\begin{proposition}\label{prop:Partner-Unique-Globaly-Stable-General}
    If $v_{x^*}(a_d|\cdot)=u_{x^*}(a_d|\cdot)$ whenever $x^*(a_d)=0$ then the generalised partner dynamics has a unique globally stable equilibrium. 
\end{proposition}
\begin{proof}
The proof generalises our stability result for the partner dynamics presented in Appendix \ref{app:StabilityPartner}, and explains how the proof in Appendix \ref{app:StabilityPartner} can be amended to apply for the general case. Notice that under the assumptions stated in Proposition \ref{prop:Partner-Unique-Globaly-Stable-General}, the partner dynamics of Equation (\ref{eq:GeneralPartnerDynamics}) takes the form 
\begin{align}
\label{eq:GeneralPartnerDynamicsNew}
    \dot{y}_{a_d}(a_{-d}) = (&1-r)g_1(a_d,a_{-d}|x^*)y_{a_d}(a_{-d}) \\
    &+r\biggl(\prod\limits_{d'\neq d}x^*(a_{d'})\biggr)u_{x^*}(a_d|y_{a_d})- y_{a_d}(a_{-d})u_{x^*}(a_d|y_{a_d}).
    \nonumber
\end{align}
Recall that 
\begin{gather}\label{eq:Linear-Combination-Generalised}
    u_{x^*}(a_d|y_{a_d})=\sum\limits_{a_{-d}'\in \supp_{-d}(x^*)}g_1(a_d,a_{-d}'|x^*)y_{a_d}(a_{-d}').
\end{gather} 
In Appendix \ref{app:StabilityPartner}, we show $g_1(a|x^*)=u_{x^*}(a)$ (following a normalisation of units by factor $u_{x^*}$), 
and substituting this form into Equation (\ref{eq:GeneralPartnerDynamicsNew}) yields the partner dynamics of Equation (\ref{eq:ApproxPartnerDynamics}). 


The proof in Appendix \ref{app:StabilityPartner} immediately applies to our setting here by replacing $u_{x^*}(a)$ by $g_1(a|x^*)$. 
For uniqueness, Proposition \ref{prop:UniquePartner} applies to our setting here by, again, replacing $u_x(a)$ by $g_1(a|x)$
\end{proof}

  Next, we show that our generalised results allow us to capture dynamics akin to imitation by success:

       \begin{example}\label{exm:TraitImitationLinear}
       Consider the linear functional form $G(u)=\prod\limits_{d\in D}u_d+b$.
       The dynamics in this case is given by 
       \begin{align}\label{eq:b-recombinator-dynamic}
           \dot{x}(a) &=(1-r)\left[(1+b)^{|D|-1}\left(\frac{u_x(a)}{u_x}+b\right)-(1+b)^{|D|}+1\right]x(a)\\
           \nonumber &+r\prod\limits_{d\in D} x(a_d)\left[\prod\limits_{d\in D}\left(\frac{u_x(a_d)}{u_x}+b\right)-(1+b)^{|D|}+1\right]-x(a)
       \end{align}
       It thus holds that $u_x(a_d|\cdot)=v_x(a_d|\cdot)$ and thus Proposition \ref{prop:Partner-Unique-Globaly-Stable-General} applies. 
       
       In terms of the literature, this example should be compared with Example $5.4.4$ in \citealp{sandholm2010population}
       (imitation by success), where (using Sandholm's notation) $\rho_{ij}(x,\pi)=x_ic(\pi_j)$ and $c(\pi_j)=\pi_j+K$. 
       In that case, the dynamics leads to the replicator dynamics. Note that if $|D|=1$ then in our case
       \begin{gather}
       \label{eq:D=1}
           \dot{x}(a)=(1-r)\frac{u_x(a)}{u_x}x(a)+r\frac{u_x(a)}{u_x}x(a)-x(a)=\left(\frac{u_x(a)}{u_x}-1\right)x(a)
       \end{gather}
       hence we also derive the replicator dynamics for this special case. However, when $|D|>1$ we derive substantially different dynamics. For example, when  $|D|=2$ we get the following dynamic, 
       \begin{align}
       \label{eq:b-recombinator-dynamic-D=2}
           \dot{x}(a)=
           (1-r)(1+b)\left(\frac{u_x(a)}{u_x}+b\right)x(a)+r\prod\limits_{d\in D} x(a_d)\prod\limits_{d\in D}\left(\frac{u_x(a_d)}{u_x}+b\right)\\ \nonumber -(1+b)^2\left((1-r)x(a)+r \prod\limits_{d\in D} x(a_d)\right)-r\left(x(a)-\prod\limits_{d\in D} x(a_d)\right),
       \end{align}
       significantly different from the replicator of Equation (\ref{eq:D=1}).
\hfill \myenddefinition

\end{example}

\end{document}